\documentclass[conference]{ieeeconf}
\IEEEoverridecommandlockouts
\usepackage{cite}
\usepackage{amsmath,amssymb,amsfonts}
\usepackage{graphicx}
\usepackage{textcomp}
\usepackage{xcolor}
\def\BibTeX{{\rm B\kern-.05em{\sc i\kern-.025em b}\kern-.08em
    T\kern-.1667em\lower.7ex\hbox{E}\kern-.125emX}}

\usepackage{graphicx}
\graphicspath{{../figures/}}
\usepackage{hhline}
\usepackage{caption,subcaption}
\usepackage{url}
\usepackage{enumerate}

\usepackage{algorithm}
\usepackage{algpseudocode}
\usepackage{algorithmicx}

\newcommand{\elis}[1]{\normalsize{\color{red}(Elis:\ #1)}}

\newcommand{\M}{\mathcal{M}}

\newcommand{\RR}{\mathbb{R}}

\newcommand{\ZZ}{\mathbb{Z}}

\newcommand{\depth}{\mathrm{depth}}

\newtheorem{definition}{Definition}
\newtheorem{theorem}{Theorem}
\newtheorem{lemma}{Lemma}

\newtheorem{remark}{Remark}

\newtheorem{proposition}{Proposition}

\usepackage{hyperref}

\def\longversion{1} 

\usepackage{tikz}
\usetikzlibrary{shapes.geometric} 
\usetikzlibrary{calc}
\usetikzlibrary{fit}
\usetikzlibrary{arrows,automata}

\tikzstyle{process} = [rectangle, minimum width=3cm, minimum height=1cm, text centered, text width=3cm, draw=black]
\tikzstyle{decision} = [diamond, minimum width=3cm, minimum height=1cm, text badly centered, draw=black]
\tikzstyle{cloud} = [draw, ellipse, node distance=3cm, minimum height=2em]
\tikzstyle{startstop} = [rectangle, rounded corners, minimum width=1cm, minimum height=1cm,text centered, draw=black]
\tikzstyle{arrow} = [thick,->,>=stealth]

\tikzstyle{block} = [draw, rectangle, 
    minimum height=3em, minimum width=2.9cm, text centered]
\tikzstyle{sum} = [draw, fill=blue!20, circle, node distance=1cm]
\tikzstyle{input} = [coordinate]
\tikzstyle{output} = [coordinate]
\tikzstyle{pinstyle} = [pin edge={to-,thin,black}]

\begin{document}

\title{\LARGE \bf
Hierarchical Finite State Machines for Efficient Optimal Planning in Large-scale Systems
}


\author{Elis Stefansson$^{1}$ and Karl H. Johansson$^{1}$
\thanks{$^{1}$School of Electrical Engineering and Computer Science, KTH Royal Institute of Technology, Sweden. Email:
        {\tt\small \{elisst,kallej\}@kth.se}. The authors are also affiliated with Digital Futures.}%
\thanks{This work was partially funded by the Swedish Foundation for Strategic Research, the Swedish Research Council, and the Knut och Alice Wallenberg~foundation.}
}

\maketitle

\maketitle

\begin{abstract}
In this paper, we consider a planning problem for a hierarchical finite state machine (HFSM) and develop an algorithm for efficiently computing optimal plans between any two states. The algorithm consists of an offline and an online step. In the offline step, one computes exit costs for each machine in the HFSM. It needs to be done only once for a given HFSM, and it is shown to have time complexity scaling linearly with the number of machines in the HFSM. In the online step, one computes an optimal plan from an initial state to a goal state, by first reducing the HFSM (using the exit costs), computing an optimal trajectory for the reduced HFSM, and then expand this trajectory to an optimal plan for the original HFSM. The time complexity is near-linearly with the depth of the HFSM. It is argued that HFSMs arise naturally for large-scale control systems, exemplified by an application where a robot moves between houses to complete tasks. We compare our algorithm with Dijkstra's algorithm on HFSMs consisting of up to 2 million states, where our algorithm outperforms the latter, being several orders of magnitude faster.








\end{abstract}


\section{Introduction}

\subsection{Motivation}\label{motivation}

Large-scale control systems are becoming ubiquitous as we move towards smarter and more connected societies. Therefore, analysing and optimising the performance of such systems is of outmost importance. One common approach to facilitate the analysis of a large-scale system is to break up the system into subsystems, analyse these subsystems separately, and then infer the performance of the whole system. This ideally ease the analysis and enables reconfiguration (e.g., if one subsystem is changed, the whole system does not need to be reanalysed).

One framework suited for modelling systems made of subsystems is the notion of a hierarchical finite state machine (HFSM). Originally introduced by Hashel \cite{harel1987statecharts}, an HFSM is a machine composed of several finite state machines (FSMs) nested into a hierarchy. The motivation is to conveniently model complex systems in a modular fashion being able to represent and depict subsystems and their interaction neatly.

In this paper, we are interested in how to optimally plan in HFSMs. As an example, consider a robot moving between warehouses as in Fig. \ref{fig:motivating_example}. In each warehouse, there are several locations the robot can go to, and at each location the robot can do certain tasks (e.g., scan a test tube), with decision costs given by some cost functional. This system can naturally be modelled as an HFSM with a hierarchy consisting of three layers (warehouse, location and task layer). A key question is then how to design efficient planning algorithms for such HFSMs that take into account the hierarchical structure of the system, seen as a first step towards more efficient planning algorithms for large-scale systems in general.

\subsection{Contribution}
We consider optimal planning in systems modelled as HFSMs and present an efficient algorithm for computing an optimal plan between any two states in the HFSM. More precisely, our contributions are three-fold:

Firstly, we extend the HFSM formalism in \cite{biggar2021modular} to the case when machines in the hierarchy have costs, formalised by Mealy machines (MMs) \cite{Mealy1955}, and call the resulting hierarchical machine a hierarchical Mealy machine (HiMM).

Secondly, we present an algorithm for efficiently computing optimal plans between any two states in an HiMM. The algorithm consists of an offline step and an online step. In the offline step, one computes exit costs for each MM in the HiMM. It needs to be done only once for a given HiMM, and it is shown to have time complexity scaling linearly with the number of machines in the HiMM, able to handle large systems. In the online step, one computes an optimal plan from an initial state to a goal state, by first reducing the HiMM (using the exit costs), computing an optimal trajectory to the reduced HiMM, and then expand this trajectory to an optimal plan for the original HiMM. The partition into an offline and online step enables rapid computations of optimal plans by the online step. Indeed, it is shown that the online step obtains an optimal trajectory to the reduced HiMM in time $O(\mathrm{depth}(Z) \log ( \mathrm{depth}(Z) ) )$, where $\mathrm{depth}(Z)$ is the depth of the hierarchy of the considered HiMM $Z$, and can then use this trajectory to retrieve the next optimal input of the original HiMM in time $O(\mathrm{depth}(Z))$, or obtain the full optimal plan $u$ at once in time $O(\mathrm{depth}(Z) |u|)$, where $|u|$ is the length of $u$. This should be compared with Dijkstra's algorithm which could be more than exponential in $\mathrm{depth}(Z)$ \cite{Dijkstra1959, DijkstraFibonacci}.

Thirdly, we show-case our algorithm on the robot application introduced in the motivation and validate it on large hierarchical systems consisting of up to 2 million states, comparing our algorithm with Dijkstra's algorithm. Our algorithm outperforms the latter, where the partition into an offline and online step reduces the overall computing time, and the online step computes optimal plans in just milliseconds compared to tens of seconds using Dijkstra's~algorithm.


\begin{figure}[t]
	\centering
  \includegraphics[width=0.4\textwidth]{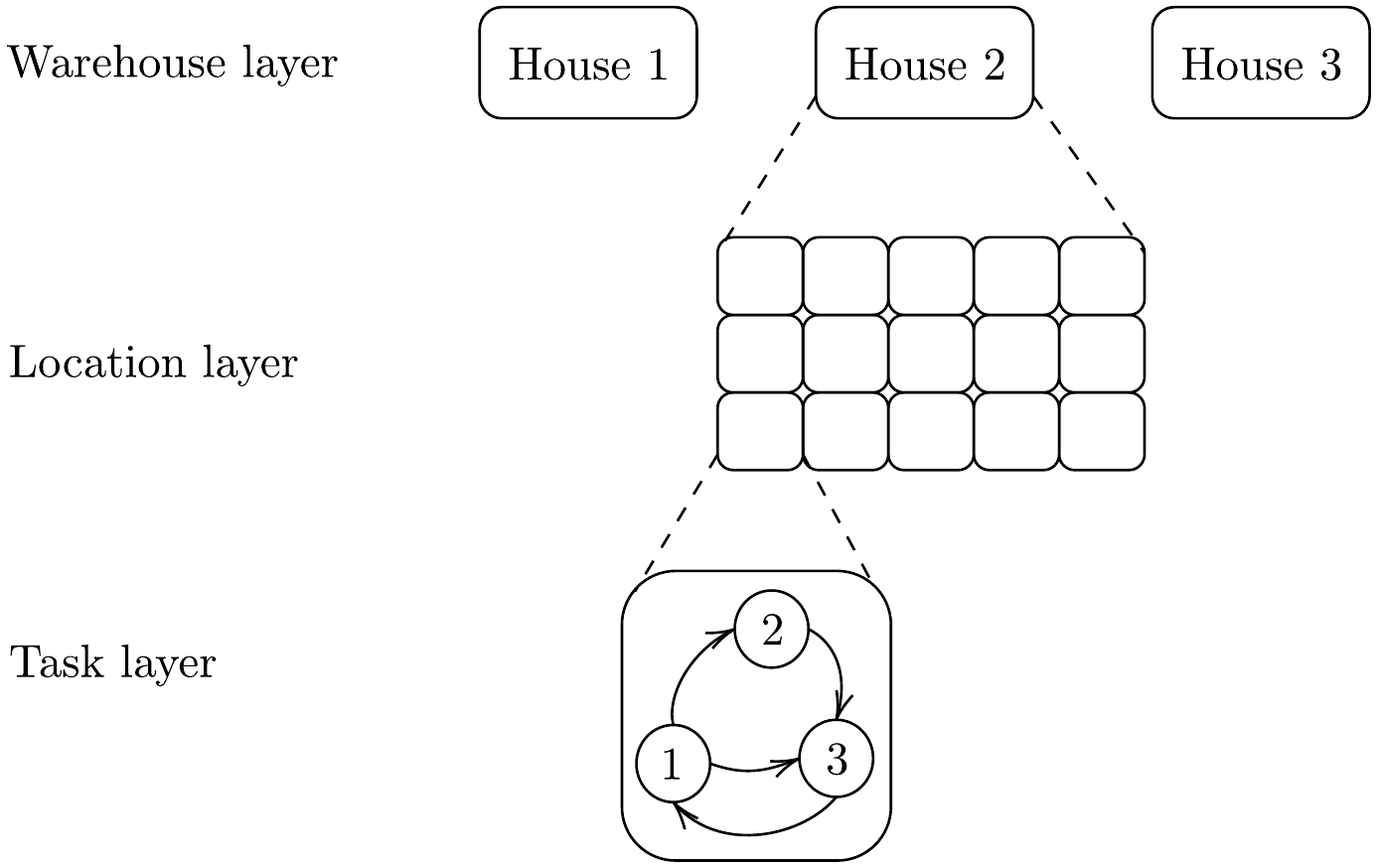}
  \caption{A mobile robot moving between warehouses modelled as a three-layer HFSM.} 
  \label{fig:motivating_example}
\end{figure}

\subsection{Related Work}
Traditionally, HFSMs has been used to model reactive agents, such as wrist-watches \cite{harel1987statecharts}, rescue robots \cite{schillinger2016human} and non-player characters in games \cite{millington2018artificial}. Here, the response of the agent (the control law) is represented as an HFSM reacting to inputs from the environment (e.g., ``hungry''), where subsystems typically correspond to subtasks (e.g, ``get food''). This paper differs from this line of work by instead treating the environment as an HFSM where the agent can \emph{choose} the inputs fed into the HFSM (i.e., inputs are now decision variables), with aim to steer the system to a desirable~state. In discrete event systems \cite{cassandras2008introduction}, a variant of HFSMs known as state tree structures has been used to compute safe executions \cite{ma2006nonblocking, wang2020real}. We differ in the HFSM formalism and focus instead on optimal planning with respect to a cost functional.

There is an extensive literature when it comes to path planning in discrete systems \cite{bast2016route,lavalle2006planning}. Hierarchical path planning algorithms, e.g., \cite{dibbelt2016customizable,mohring2007partitioning,10.1145/1498698.1564502}, are the ones most reminiscent to our approach due to their hierarchical structure. Related algorithms consider path planning on weighted graphs, pre-arranging the graph into clusters to speed up the search, and could be used to plan FSMs. However, to apply such methods for an HFSM, we would first need to flatten the HFSM to an equivalent flat FSM, making the algorithm agnostic to the modular structure of the HFSM (and thus less suitable for instance reconfigurations in the HFSM), and could also in the worst case (when reusing identical components of the HFSM, see \cite{alur1998model, yannakakis2000hierarchical} for details) cause an exponential time complexity. It is therefore beneficial to instead consider path planning in the HFSM directly. This is done in this work. The work \cite{timo2014reachability} seeks an execution of minimal length between two configurations in a variant of an HFSM. This paper differ in the HFSM formalism and consider non-negative transition costs instead of just minimal length.






Finally, the work \cite{biggar2021modular} formalises HFSMs without outputs, and uses a modular decomposition to decompose an FSM into an equivalent HFSM. Planning is not considered. In this work, we extend the HFSM formalism from \cite{biggar2021modular} to HFSMs with outputs (modelling costs) and consider optimal~planning.

\subsection{Outline}
The outline of the paper is as follows. Section \ref{problem_formulation} formally defines HiMMs and formulates the problem statement. Section \ref{hierarchical_planning} presents our planning algorithm. Section \ref{numerical_evaluations} validate our algorithm in numerical evaluations. Finally, Section \ref{conclusion} concludes the paper. \if\longversion0 An extended version of this paper can be found at \cite{stefansson2023ecc} that contains all the proofs in Appendix.\else
All proofs can be found in Appendix.
\fi

\section{Problem Formulation}\label{problem_formulation}

\subsection{Hierarchical Finite State Machines}
We follow the formalism of \cite{biggar2021modular} closely when defining our hierarchical machines, extending their setup to the case when machines also have outputs. Formally, we consider Mealy machines \cite{Mealy1955} and then define hierarchical Mealy~machines.

\begin{definition}[Mealy Machine]\label{FSM_def}
An MM is a tuple $M = (Q,\Sigma,\Lambda,\delta,\gamma,s)$, where $Q$ is a finite set of states; $\Sigma$ is a finite set of inputs, the input alphabet; $\Lambda$ is a finite set of outputs, the output alphabet; $\delta: Q \times \Sigma \rightharpoonup Q$ is the transition function, which can be a partial function\footnote{We use the notation $f : A \rightharpoonup B$ to denote a partial function from a set $A$ to a set $B$ (i.e., a function that is only defined on a subset of A). If $f(a)$ with $a \in A$ is not defined, then we write $f(a) =\emptyset$.};  $\gamma: Q \times \Sigma \rightarrow \Lambda$ is the output function; and $s \in Q$ is the start state.
\end{definition}
An MM $M$ works as follows. When initialised, $M$ starts in the start state $s$. Next, given a current state $q \in Q$ and input $x \in \Sigma$, $M$ outputs $\gamma(q,x) \in \Lambda$, and transits to the state $\delta(q,x)$ if $\delta(q,x) \neq \emptyset$ (i.e., if $\delta(q,x)$ is defined), otherwise $M$ stops. Repeating this process results in a trajectory of $M$:
\begin{definition}[Trajectory]
A sequence $z = \{(q_i,x_i)\}_{i=1}^N$ with $N \in \ZZ^+$ is a \emph{trajectory} of an MM $M = (Q,\Sigma,\Lambda,\delta,\gamma,s)$ (starting at $q_1 \in Q$) if $q_{i+1} = \delta(q_i,x_i) \neq \emptyset$ for $i \in \{1,\dots,N-1 \}$.
\end{definition}
In this work, we assume that we can \emph{choose} the inputs. In such settings, we also talk about plans and their corresponding induced trajectories. More precisely, a \emph{plan} is a sequence $u = (x_i)_{i=1}^N \in \Sigma^N$ with $N \in \ZZ^+$, and we call $z = \{(q_i,x_i)\}_{i=1}^N$ the \emph{induced trajectory} to $u$ starting at $q_1 \in Q$ if $z$ is a trajectory.
Finally, we sometimes use the notation $Q(M)$, $\Sigma(M)$, $\Lambda(M)$, $\delta_M$, $\gamma_M$ and $s(M)$ to stress that e.g., $Q(M)$ is the set of states of~$M$.


\begin{remark}
Here, $\delta$ is a partial function to model stops in the machine. In the hierarchical setup, this allows higher-layer machines in the hierarchy to be called when lower-layer machines are completed. See \cite{biggar2021modular} for a detailed account.
\end{remark}


\begin{definition}[Hierarchical Mealy Machine]\label{HFSM_output_def}
An HiMM is a pair $Z = (X,T)$, where $X$ is a set of MMs with input set $\Sigma$ and output set $\Lambda$ (the MMs in $Z$), and $T$ is a tree with the MMs in $X$ as nodes (specifying how the MMs in $X$ are composed in $Z$). More precisely, each node $M \in X$ in $T$ has $|Q(M)|$ labelled outgoing arcs $\{M \xrightarrow{q} M_q \}_{q \in Q(M)}$, where either $M_q \in X$ (meaning that state $q$ of $M$ corresponds to the MM $M_q$ one layer below in the hierarchy of $Z$) or $M_q = \emptyset$ (meaning that $q$ is just a state without refinement). For brevity, call $Q_Z := \cup_{M \in X} Q(M)$ the nodes of $Z$ and $S_Z := Q_Z \cap \{q: M_q = \emptyset \}$ the states of $Z$. The depth of $Z$, $\mathrm{depth}(Z)$, is the depth of the tree $T$, i.e., the maximum over all directed path lengths in $T$. Furthermore, we also have notions of the start state, the transition function and the output function (where $(X_i \xrightarrow{v} X_j) \in T$ means that there is an arc labelled $v$ from $X_i$ to $X_j$ in $T$):
\begin{enumerate}[(i)]
\item Start function: The function $\mathrm{start}: X \rightarrow S_Z$ is
\begin{equation*}
\mathrm{start}(X_i) =
\begin{cases}
\mathrm{start}(X_j), &  \textrm{$(X_i \xrightarrow{s(X_i)} X_j) \in T$}, X_j \in X \\
s(X_i), & \textrm{otherwise.}
\end{cases}
\end{equation*}
\item Hierarchical transition function: Let $q \in Q(X_j)$, where $X_j \in X$, and $v = \delta(q,x)$. Then the hierarchical transition function $\psi: Q_Z \times \Sigma \rightharpoonup S_Z$ is defined as
\begin{equation*}
\psi(q,x) =
\begin{cases} 
\mathrm{start}(Y), & v \neq \emptyset, (X_j \xrightarrow{v} Y) \in T, Y \in X \\
v, & v \neq \emptyset, \mathrm{otherwise} \\
\psi(w,x), & v = \emptyset, (W \xrightarrow{w} X_j) \in T, W \in X \\
\emptyset, & v = \emptyset, \mathrm{otherwise.}
\end{cases}
\end{equation*}
\item Hierarchical output function: Let $q \in Q(X_j)$ where $X_j \in X$ and $v = \delta(q,x)$. Then the hierarchical output function $\chi: Q_Z \times \Sigma \rightharpoonup \Lambda$ of $Z$ is defined as
\begin{equation*}
\chi(q,x) =
\begin{cases} 
\gamma_{X_j}(q,x), & v \neq \emptyset, (X_j \xrightarrow{v} Y) \in T, Y \in X \\
\gamma_{X_j}(q,x), & v \neq \emptyset, \mathrm{otherwise} \\
\chi(w,x), & v = \emptyset, (W \xrightarrow{w} X_j) \in T, W \in X \\
\emptyset, & v = \emptyset, \mathrm{otherwise.}
\end{cases}
\end{equation*}
\end{enumerate}
\end{definition}

An HiMM $Z = (X,T)$ works analogously to an MM. When initialised, the HiMM starts at state $\mathrm{start}(M_0) \in S_Z$ (where $M_0$ is the root of $T$). Next, given a current state $q \in S_Z$ and input $x \in \Sigma$, $Z$ outputs $\chi(q,x) \in \Lambda$, and transits to the state $\psi(q,x) \in S_Z$ if $\psi(q,x) \neq \emptyset$, otherwise $Z$ stops. Furthermore, a trajectory, plan, and induced trajectory are (with obvious modifications) defined totally analogously as for an~MM. 


\begin{figure}[t]
\centering
\includegraphics[width=1\linewidth]{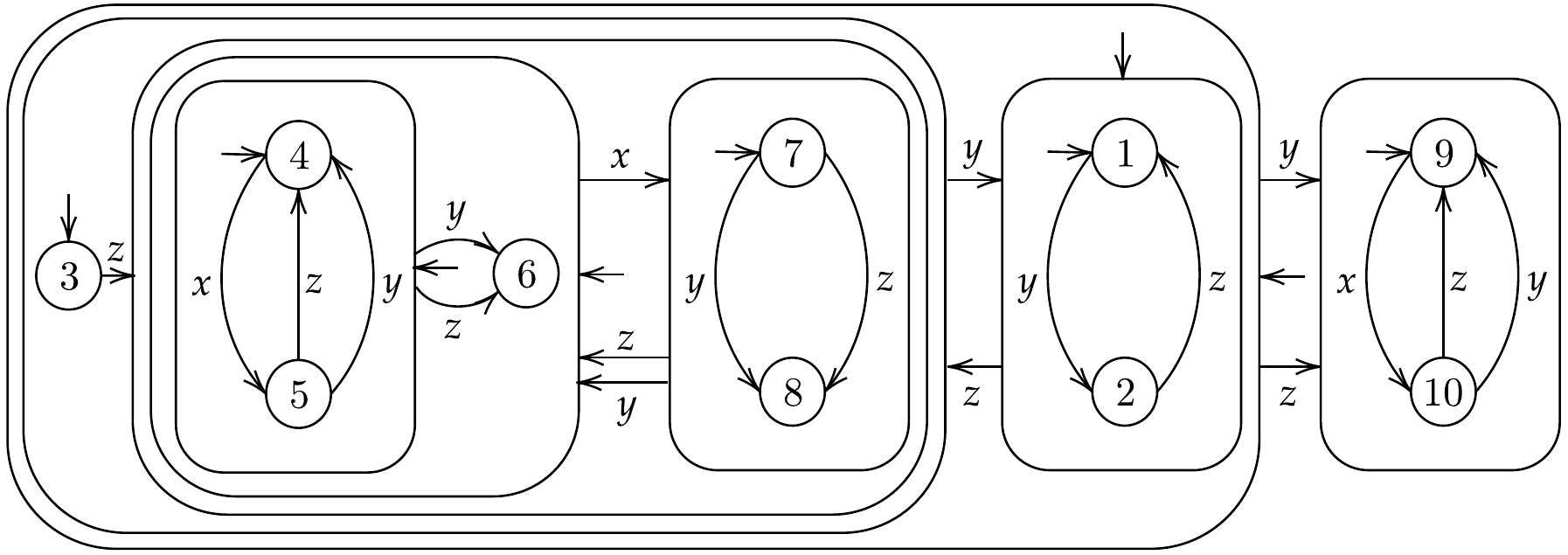}
\caption{Example of an HiMM taken from \cite{biggar2021modular} but with added outputs, assumed to be unit costs.}
\label{fig:HFSM_example}
\end{figure}

\begin{remark}[Intuition]
For intuition regarding Definition \ref{HFSM_output_def}, consider the HiMM $Z = (X,T)$ in Fig. \ref{fig:HFSM_example}. The HiMM is from \cite{biggar2021modular} but with added outputs, for simplicity assumed to be all unit costs, hence, we omit writing them. Here, the states of $Z$ are the circles labelled 1 to 10, while the nodes of $Z$ are the circles plus the (unlabelled) rounded rectangles, which we denote by the states they contain, e.g., $\{9,10\}$. Furthermore, the inputs are $\Sigma = \{x,y,z\}$, where labelled arrows denote corresponding transitions, and small arrows specify start states. For example, in the MM having nodes $\{4,5\}$ and $6$ as states, one starts in $\{4,5 \}$ and can e.g., transition from $\{4,5 \}$ to $6$ with input $y$. To get intuition concerning the hierarchical transition function $\psi$, consider the case when $Z$ is in state 2 and apply input $y$. If there would have been a $y$-transition from 2, then we would have just moved according to that transition. However, this is not the case and hence, we instead move up iteratively in the hierarchy until we find a node that has a $y$-transition (or stop $Z$ if we don't find one). In this case, the node $\{1,2\}$ (just above 2 in the hierarchy) does not support a $y$-transition either, but the node $\{1,\dots,8\}$ (above $\{1,2\}$) does and we therefore move according to that transition, i.e., to $\{9,10\}$. Once moved, we iteratively follow the start states down in hierarchy until we arrive at a state of $Z$, in this case from $\{9,10\}$ to $9$. With this step, the procedure is complete, that is, $\psi(2,y)=9$. Moreover, the output $\chi(2,y)$ captures the transition corresponding to the $y$-transition going from $\{1,\dots,8\}$ with $y$, i.e., $\chi(2,y) = \gamma_{M}(\{1,\dots,8\},y)$, where $M$ is the MM in $Z$ having the node $\{1,\dots,8\}$ as state. The motivation is to support modularity neatly, where we in our case care about that we moved in $M$ from $\{1,\dots,8\}$ with $y$, but not how we arrived at $\{1,\dots,8\}$ (with $y$) from layers below. Future work will consider variations of this setup.
Finally, we also depict the HiMM in a tree-like structure given by Fig. \ref{fig:hfsm_example_tree_cut}, useful when illustrating the planning algorithm given in the next section.
\end{remark}

\begin{figure}[t]
\centering
\includegraphics[width=0.9\linewidth]{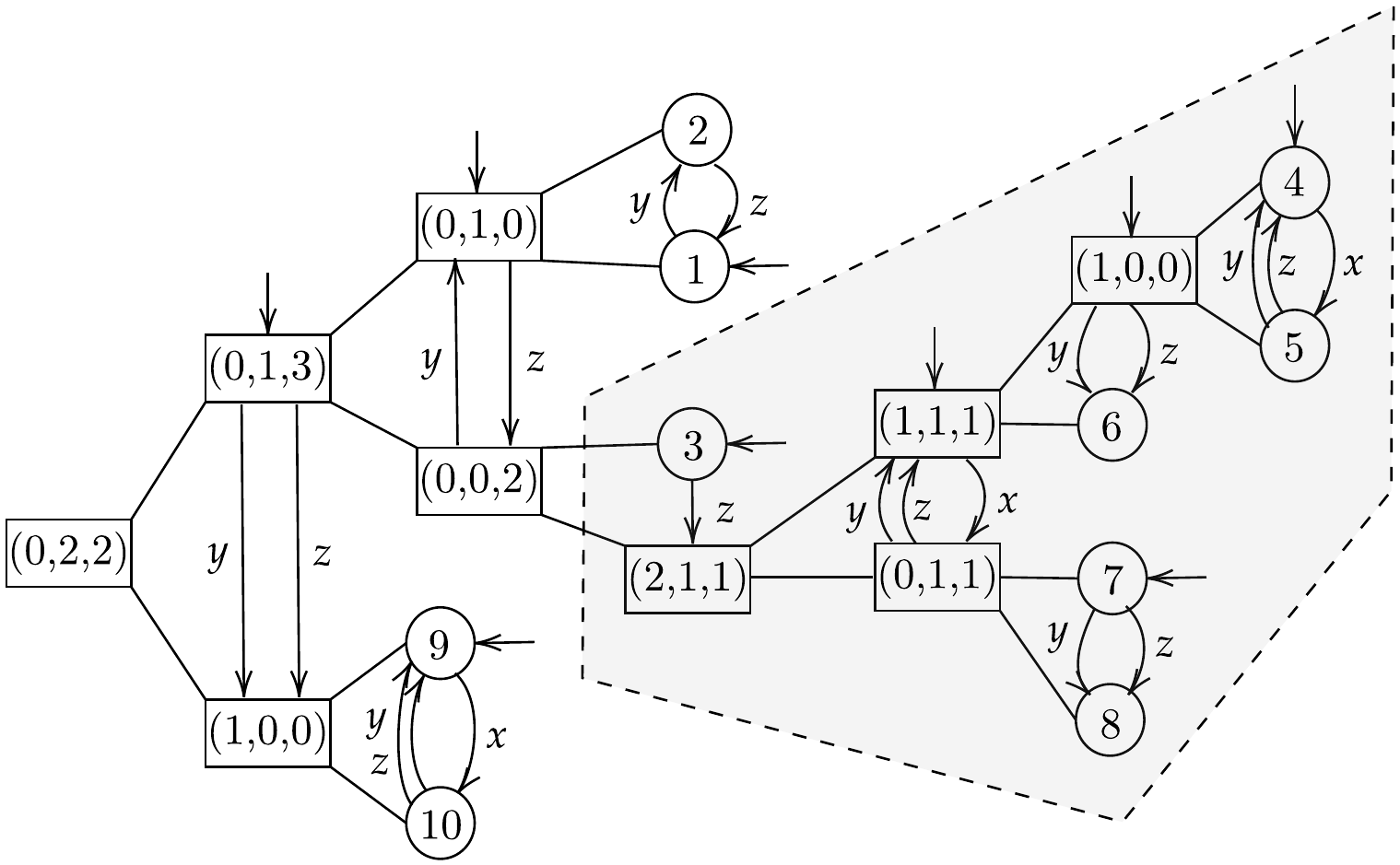}
\caption{Tree representation of the HiMM in Fig. \ref{fig:HFSM_example}. The tuple on an MM $M$ is the optimal exit costs $(c^M_x,c^M_y,c^M_z)$ obtained by the offline step. The light-grey region depicts the part that is removed in the online step with $s_{\mathrm{init}} = 2$ and $s_{\mathrm{goal}} = 10$.}
\label{fig:hfsm_example_tree_cut}
\end{figure}

In this paper, we exclusively consider HiMMs such that $\Lambda \subset \RR^+$ (henceforth implicitly assumed), interpreting $\chi(q,x)$ as the \emph{cost} for executing $x$ at node $q$. For such an HiMM $Z$, we also define a cumulative cost. More precisely, for a trajectory $z = \{(q_i,x_i)\}_{i=1}^N$ of $Z$, the cumulative cost is $C(z) = \sum_{i=1}^N \chi(q_i,x_i)$ if all $\chi(q_i,x_i) \neq \emptyset$, and $C(z) = + \infty$ otherwise.\footnote{In fact, since $z$ is a trajectory, only $\chi(q_N,x_N)$ may be empty.} We put $C(z) = + \infty$ in the latter case since we do not want $Z$ to stop.

\subsection{Problem Statement}
We now formalise the problem statement. Let $Z = (X,T)$ be an HiMM such that $\Lambda \subset \RR^+$. Consider state $s_{\mathrm{init}} \in S_Z$, called the \emph{initial state}, and $s_{\mathrm{goal}} \in S_Z$ called the \emph{goal state} (any states of $Z$). Let $\Psi(s_\mathrm{init},s_\mathrm{goal})$ be the set of plans $u =(x_i)_{i=1}^N$ that ends at $s_\mathrm{goal}$ if starting at $s_\mathrm{init}$, i.e., the induced trajectory $z = \{(q_i,x_i)\}_{i=1}^N$ to $u$ starting at $q_1 = s_\mathrm{init}$ exists (i.e., $z$ is a trajectory) and $\psi(q_N,x_N)= s_\mathrm{goal}$. Then, find a plan $\hat{u} \in \Psi(s_\mathrm{init},s_\mathrm{goal})$ that minimises the cumulative cost
\begin{equation*}
\min_{u \in \Psi(s_\mathrm{init},s_\mathrm{goal})} C(z),
\end{equation*} 
where $z$ is the induced trajectory to $u$ starting at $s_\mathrm{init}$. We call such a $\hat{u}$ an \emph{optimal plan} to the \emph{planning objective} $(Z,s_\mathrm{init},s_\mathrm{goal})$, and the induced trajectory $\hat{z}$ an \emph{optimal trajectory} to $(Z,s_\mathrm{init},s_\mathrm{goal})$. Moreover, a plan $u \in \Psi(s_\mathrm{init},s_\mathrm{goal})$ is called a \emph{feasible plan} to $(Z,s_\mathrm{init},s_\mathrm{goal})$, and we say that $(Z,s_\mathrm{init},s_\mathrm{goal})$ is \emph{feasible} if there exists a feasible plan.

\section{Hierarchical planning}\label{hierarchical_planning}
In this section, we present our hierarchical planning algorithm, providing an overview in Section \ref{hierarchical_planning_overview}, followed by details in Sections \ref{optimal_exit_table} to \ref{online_step_algorithm}. We fix an arbitrary planning objective $(Z,s_\mathrm{init},s_\mathrm{goal})$ throughout this section.

\subsection{Overview}\label{hierarchical_planning_overview}
The hierarchical planning algorithm finds an optimal plan $u$ to $(Z,s_\mathrm{init},s_\mathrm{goal})$. The algorithm is summarised by Algorithm \ref{alg:hierarchical_planning} and consists of an offline step and an online step. The offline step computes optimal exit costs $(c_x^M)_{x \in \Sigma}$ and corresponding trajectories $(z_x^M)_{x \in \Sigma}$ for each MM $M \in X$ (line 2). This step needs to be done only once for a given HiMM $Z$. The online step then computes an optimal plan $u$ to $(Z,s_\mathrm{init},s_\mathrm{goal})$ using the result from the offline step. More precisely, it first reduces $Z$ to an equivalent reduced HiMM $\bar{Z}$ (equivalence given by Theorem \ref{theorem:planning_equivalence}), pruning all irrelevant MMs of $Z$ and replacing them with corresponding costs from the offline step, and then finds an optimal trajectory $z$ to $\bar{Z}$ (line 4). Then, it expands $z$ to an optimal trajectory for $(Z,s_\mathrm{init},s_\mathrm{goal})$ from which we get the optimal plan $u$ (line 5). The details of the offline and online step is given by Section \ref{optimal_exit_table} and Sections \ref{online_step} to \ref{online_step_algorithm}, respectively, where lines 2, 4 and 5 in Algorithm \ref{alg:hierarchical_planning} correspond to Algorithm \ref{alg:optimal_exit_table}, \ref{alg:Step1} and \ref{alg:step_3}, respectively.

\begin{algorithm}[t]
\caption{Hierarchical planning}\label{alg:hierarchical_planning}
\begin{algorithmic}[1]
\Require HiMM$ \; Z =(X,T)$ and states $s_{\mathrm{init}}, s_{\mathrm{goal}}$.
\Ensure Optimal plan $u$ to $(Z,s_{\mathrm{init}}, s_{\mathrm{goal}})$.
\State \textbf{Offline step:}
\State $(c_x^M,z_x^M)_{x \in \Sigma, M \in X} \gets \mathrm{Offline\_step}(Z)$
\State \textbf{Online step:}
\State $z \gets \mathrm{Reduce\_and\_solve}(Z,s_{\mathrm{init}},s_{\mathrm{goal}},(c_x^M,z_x^M)_{x \in \Sigma, M \in X})$
\State $u \gets \mathrm{Expand}(z,(z_x^M)_{x \in \Sigma, M \in X}, Z)$
\end{algorithmic}
\end{algorithm}

\subsection{Offline Step}\label{optimal_exit_table}
Towards a precise formulation of the offline step, we need the following notions.
A state $q \in S_Z$ is \emph{contained} in an MM $M \in X$ if $q$ is a descendant of $M$ in $T$ (e.g., state 4 is contained in the MM labelled (2,1,1) in Fig. \ref{fig:hfsm_example_tree_cut}). A trajectory $z = \{(q_i,x_i)\}_{i=1}^N$ is an \emph{$(M,x)$-exit trajectory} if every $q_i$ is contained in $M$, $x_N = x$, and $q_{N+1} := \psi(q_N,x_N)$ (possibly empty) is not contained in $M$ (hence, $z$ exits $M$ with $x$). The corresponding \emph{$(M,x)$-exit cost} of $z$ equals $\sum_{i=1}^{N-1} \chi(q_i,x_i)$ (the cost of $(q_N,x_N)$ is excluded since the transition goes outside the subtree with root $M$). The \emph{optimal $(M,x)$-exit cost}, denoted $c_x^M$, is the minimal $(M,x)$-exit cost,\footnote{If no $(M,x)$-exit trajectories exists, then $c_x^M = \infty$ and any trajectory contained in $M$ is said to be an optimal $(M,x)$-exit~trajectory.} and any such $c_x^M$ is called an \emph{optimal exit cost} of~$M$. An $(M,x)$-exit trajectory that achieves the optimal $(M,x)$-exit cost is an \emph{optimal $(M,x)$-exit~trajectory}, and any such trajectory is called an \emph{optimal exit trajectory} of $M$.


We now provide the details of the offline step. The offline step obtains $(c_x^M)_{x \in \Sigma}$ for each MM $M$ in $Z =(X,T)$ recursively over the tree $T$. More precisely, let $M = (Q,\Sigma,\Lambda,\delta,\gamma,s)$ be an MM of $Z$. To compute $(c_x^M)_{x \in \Sigma}$, form the augmented MM $\hat{M}$ given by $\hat{M} = (Q \cup \{ E_x \}_{x \in \Sigma},\Sigma,\Lambda,\hat{\delta},\hat{\gamma},s)$. Here, $\hat{M}$ is identical to $M$ except that whenever an input $x \in \Sigma$ would exit $M$ ($\delta(q,x)=\emptyset$) then one instead goes to the added state $E_x$ in $\hat{M}$. That is, $\hat{\delta}(q,x) := \delta(q,x)$ if $\delta(q,x) \neq \emptyset$ and $\hat{\delta}(q,x) = E_x$ otherwise (the values of $\hat{\delta}(q,x)$ for $q \in \{ E_x \}_{x \in \Sigma}$ are immaterial). Furthermore, $\hat{\gamma}$ is given by $\hat{\gamma}(q,x) := c_x^{q}+\gamma(q,x)$ if $\delta(q,x) \neq \emptyset$ and $\hat{\gamma}(q,x) := c_x^{q}$ otherwise (again, the values of $\hat{\gamma}(q,x)$ for $q \in \{ E_x \}_{x \in \Sigma}$ are immaterial). Here, $c^{q}_x = 0$ if $q \in S_Z$ is a state of $Z$, and $c^{q}_x = c^{M_{q}}_x$ otherwise, where $M_q$ is the MM corresponding to $q$. Thus, intuitively, $\hat{\gamma}(q,x)$ reflects the cost of exiting $q$ with $x$ plus the cost of applying $x$ from $q$ in $M$. Note also that, by recursion (going upwards in the tree $T$), we may assume that all $c_x^{q}$ are already known.

\begin{algorithm}[t]
\caption{Offline\_step}\label{alg:optimal_exit_table}
\begin{algorithmic}[1]
\Require HiMM $Z = (X,T)$ with markings.
\Ensure Computed $(c_x^M,z_x^M)_{x \in \Sigma}$ for each MM $M$ of $Z$
\State $\mathrm{Optimal\_exit}(M_0)$ \Comment{Run from root MM $M_0$ of $T$}
\State $\mathrm{Optimal\_exit}(M)$: \Comment{Recursive help function} 
\For {each state $q$ in $Q(M)$}
\If{$q \in S_Z$}
\State $(c_x^q)_{x \in \Sigma} \gets 0_{|\Sigma|}$
\Else
\State Let $M_q$ be the MM corresponding to $q$.
\State $(c_x^{q},z_x^{q})_{x \in \Sigma} \gets \mathrm{Optimal\_exit}(M_q)$
\EndIf
\EndFor
\State Construct $\hat{M}$
\State $(c_x^M,z_x^M)_{x \in \Sigma}$ $\gets$ Dijkstra($s(M),\{E_x\}_{x \in \Sigma},\hat{M}$)
\State return $(c_x^M,z_x^M)_{x \in \Sigma}$
\end{algorithmic}
\end{algorithm}





We can now obtain $(c_x^M)_{x \in \Sigma}$ by doing a shortest path search in $\hat{M}$ starting from $s$. We use Dijkstra's algorithm \cite{Dijkstra1959, DijkstraFibonacci} computing a shortest-path tree from $s$ until we have reached all $E_x$, and thereby obtained $(c_x^M)_{x \in \Sigma}$ (with $c_x^M = \infty$ if we never reach $E_x$). We also save the corresponding trajectories $(z_x^M)_{x \in \Sigma}$ in $\hat{M}$ that we obtain for free from Dijkstra's algorithm. This procedure is performed recursively over the whole tree $T$ to obtain $(c_x^M)_{x \in \Sigma}$ and $(z_x^M)_{x \in \Sigma}$ for all MMs $M$ of $Z$. The algorithm is given by Algorithm \ref{alg:optimal_exit_table}, with correctness given by Proposition \ref{th:offline_optimal_cost} and time complexity given by Proposition \ref{th:offline_time_complexity}. See Fig.~\ref{fig:hfsm_example_tree_cut} for an example concerning the optimal exit costs.

\begin{proposition}\label{th:offline_optimal_cost}
$c_x^M \in [0,\infty]$ computed by Algorithm \ref{alg:optimal_exit_table} equals the optimal $(M,x)$-exit cost of $Z$.
\end{proposition}

\begin{proposition}\label{th:offline_time_complexity}
The time complexity of Algorithm \ref{alg:optimal_exit_table} is
\begin{equation*}
O(N [b_s |\Sigma|+(b_s+|\Sigma|) \log(b_s+|\Sigma|)]),
\end{equation*}
where $b_s$ is the maximum number of states in an MM of $Z$, and $N$ is the number of MMs in $Z$.
\end{proposition}
\begin{remark}
We stress that all the time complexity results in this section (Section \ref{hierarchical_planning}) are based on using Dijkstra's algorithm with a Fibonacci heap, due to the low time complexity, see \cite{DijkstraFibonacci} for details. However, for the systems in the simulations in Section \ref{numerical_evaluations}, we use Dijkstra's algorithm with an ordinary priority queue (that has a slightly higher time complexity) since it is in practice faster for those systems.
\end{remark}

Finally, in the remaining sections, for brevity, let a $(q,x)$-exit trajectory (cost) mean an $(M,x)$-exit trajectory (cost) if $q$ corresponds to the MM $M$. If $q$ does not correspond to an MM, i.e., $q$ is a state of $Z$, then let the $(q,x)$-exit trajectory and cost be simply $(q,x)$ and zero cost, with intuition that one can then only exit $q$ with $x$ by applying $x$ directly. An optimal $(q,x)$-exit trajectory is defined analogously and the corresponding optimal $(q,x)$-exit cost $c^{q}_x$ is as above (readily obtained from the optimal exits costs from Algorithm~\ref{alg:optimal_exit_table}).



\subsection{Online Step: Theory}\label{online_step}
We continue with the online step providing necessary theory in this section, while the algorithm is presented in Section~\ref{online_step_algorithm}.
To this end, let $U_1,\dots,U_n$ be the path of MMs in $T$ from $s_\mathrm{init}$ to the root of $T$, that is, $s_\mathrm{init}$ is a state of the MM $U_1$, the corresponding node of $U_i$ is a state of the MM $U_{i+1}$, and $U_n$ equals the root MM of $T$. Similarly, let $D_1,\dots,D_m$ be the path of MMs of $X$ in $T$ from $s_\mathrm{goal}$ to the root of $T$ (with $s_\mathrm{goal}$ being a state of the MM $D_1$). Note that there exist indices $\alpha$ and $\beta$ such that corresponding nodes of $U_\alpha$ and $D_\beta$ are states in the same MM of $X$. For brevity, let $B := D_\beta$. To get an optimal plan, we consider only the MMs of these two paths, see Fig. \ref{fig:hfsm_example_tree_cut} for an example, formalised by the \emph{reduced HiMM} in Definition \ref{reduced_planning_himm} with equivalence to the original HiMM given by Theorem \ref{theorem:planning_equivalence}. For this result, we need the notion of a reduced trajectory and an optimal expansion. These two notions can be seen as complementary operations: the first reduces a trajectory to a subset of $T$, where the latter can be used to expand a reduced trajectory with respect to a subset of $T$ to the whole tree.

\begin{definition}[Reduced trajectory]
Let $Z = (X,T)$ be an HiMM and consider a connected subset $\M$ of tree-nodes of $T$ that includes the root of $T$. Let $z = (q_i,x_i)_{i=1}^N$ be a trajectory and define the reduced trajectory $z |_\M$ of $z$ with respect to $\M$ as follows. For any state $q \in S_Z$ in some MM $M \in X$, let $\bar{q}$ be the reduced node of $q$ with respect to $\M$; that is, $\bar{q} \in Q(\bar{M})$ contains $q$, where $\bar{M}$ is the MM in $\M$ with minimal path length to $M$ in $T$. Define $\bar{z}_i = (\bar{q}_i,x_i)$ if $\bar{q}_{i+1} \neq \bar{q}_{i}$ and $\bar{z}_i = \emptyset$ otherwise. Then $z |_\M$ is the sequence of nonempty $\bar{z}_i$.
\end{definition}

In words, $z |_\M$ equals all visible transitions seen in $\M$, where $\bar{q}$ gives us the best information of the state $q$ with respect to $\M$, and $z |_\M$ changes whenever this $\bar{q}$ changes.

Next, an \emph{optimal expansion} of a node-input pair $(q,x) \in Q_Z \times \Sigma$ of $Z$ is defined by Algorithm \ref{alg:optimal_expansion}, denote it by $E(q,x)$. An optimal expansion of a sequence $z = (q_i,x_i)_{i=1}^N$ of node-input pairs is on the form $(E(q_1,x_1),\dots,E(q_N,x_N))$.
The following result justifies the notion:

\begin{proposition}\label{proposition:optimal_expansion}
Let $Z = (X,T)$ be an HiMM and $(q,x) \in Q_Z \times \Sigma$. Then $E(q,x)$ is an optimal $(q,x)$-exit trajectory.
\end{proposition}

\begin{algorithm}[t]
\caption{Optimal\_expansion}\label{alg:optimal_expansion}
\begin{algorithmic}[1]
\Require HiMM$ \; Z =(X,T)$ and node-input pair $(q,x)$.
\Ensure An optimal expansion $z$
\State return $z \gets \mathrm{Trajectory\_expansion}(q,x)$
\State $\mathrm{Trajectory\_expansion}(q,x)$: \Comment{Recursive function} 
\If {$q \in S_Z$}
\State return $(q,x)$.
\Else
\State $(q_1,x_1),\dots,(q_m,x_m) \gets z^{M_q}_x$ \Comment{$z^{M_q}_x$ from offline step} 
\For {each $(q_i,x_i)$}
\State $z_i  \gets  \mathrm{Trajectory\_expansion}(q_i,x_i)$
\EndFor
\State return $(z_1, \dots, z_m)$
\EndIf 
\end{algorithmic}
\end{algorithm}

We now define the reduced HiMM $\bar{Z}$.

\begin{definition}[Reduced HiMM]\label{reduced_planning_himm}
Let $Z = (X,T)$, $s_\mathrm{init}$ and $s_\mathrm{goal}$ be as above. Consider the HiMM $\bar{Z} = (\bar{X},\bar{T})$ where $\bar{X} = \{\bar{U}_1,\dots,\bar{U}_n\}\cup\{\bar{D}_1,\dots,\bar{D}_n\}$ and $\bar{T}$ is equal to the subtree of $T$ consisting of the nodes $U_1,\dots,U_n$ and $D_1,\dots,D_m$ but replaced with $\bar{U}_1,\dots,\bar{U}_n$ and $\bar{D}_1,\dots,\bar{D}_n$ respectively. For brevity, let $\M = \{{U}_1,\dots,{U}_n\}\cup\{{D}_1,\dots,{D}_n\}$. The details of this construction are:
\begin{enumerate}[(i)]
\item Let $U_i = (Q,\Sigma,\Lambda, \delta,\gamma, s)$. Then $\bar{U}_i = (Q,\Sigma,\Lambda, \delta,\bar{\gamma}, s)$, where $\bar{\gamma}(q,x) =$
\begin{equation}\label{eq:bar_gamma}
\begin{cases}
\gamma(q,x), & \delta(q,x) \neq \emptyset, ({U}_i \xrightarrow{q} M) \in {T}, M \in \M \\
c_x^{q}+\gamma(q,x), & \delta(q,x) \neq \emptyset, \mathrm{otherwise} \\
0, & \delta(q,x) = \emptyset, ({U}_i \xrightarrow{q} M) \in {T}, M \in \M \\
c_x^{q}, & \delta(q,x) = \emptyset, \mathrm{otherwise.}
\end{cases}
\end{equation}
Here, $c^q_x$ is as in Section \ref{optimal_exit_table}. The MM $\bar{D}_i$ is defined analogously (replace $U_i$ and $\bar{U}_i$ with $D_i$ and $\bar{D}_i$, respectively).
\item The reduced hierarchical transition function $\bar{\psi}$ is constructed according to Definition \ref{HFSM_output_def} for the HiMM $\bar{Z}$. 
\item Let $q$ be a node in $\bar{Z} = (\bar{X},\bar{T})$ where $q \in Q({X}_j)$, ${X}_j \in \bar{X}$. Let $v = \delta(q,x)$. The reduced hierarchical output function $\bar{\chi}$ is defined as $\bar{\chi}(q,x) = \gamma_{{X}_j}(q,x)$ if $v \neq \emptyset$; $\bar{\chi}(q,x) = \gamma_{{X}_j}(q,x)+\bar{\chi}(w,x)$ if $v = \emptyset$ and $(W \xrightarrow{w} X_j) \in \bar{T}$ with $W \in \bar{X}$; and, $\bar{\chi}(q,x) = \emptyset$~otherwise.
\end{enumerate}
We call $\bar{Z}$ the \emph{reduced HiMM} (with respect to $s_\mathrm{init}$ and $s_\mathrm{goal}$).
\end{definition}
Analogous to HiMMs, the \emph{reduced cumulative cost} $\bar{C}(z)$ for a trajectory $z$ of $\bar{Z}$ is $\bar{C}(z) = \sum_{i=1}^N \bar{\chi}(q_i,x_i)$ (if $\bar{\chi}(q_N,x_N) \neq \emptyset$, and $C(z) = + \infty$ otherwise), a plan that minimises $\min_{u \in \Psi(s_\mathrm{init},s_\mathrm{goal})} \bar{C}(z)$ is an \emph{optimal plan} and the induced trajectory is an \emph{optimal trajectory} to the \emph{reduced planning~objective} $(\bar{Z},s_\mathrm{init},s_\mathrm{goal})$. We have the following key~result:

\begin{theorem}[Planning equivalence]\label{theorem:planning_equivalence}
Let $Z = (X,T)$ be an HiMM and $\M = \{U_1,\dots,U_n\}\cup\{D_1,\dots,D_m\}$. Then:
\begin{enumerate}[(i)]
\item Let $z$ be an optimal trajectory to $(Z,s_\mathrm{init}, s_\mathrm{goal})$. Then, the reduced trajectory $z |_\M$ is an optimal trajectory to~$(\bar{Z}, s_\mathrm{init}, s_\mathrm{goal})$.
\item Let $z$ be an optimal trajectory to $(\bar{Z}, s_\mathrm{init}, s_\mathrm{goal})$. Then, an optimal expansion of $z$ (expanded over $Z$) is an optimal trajectory to~$(Z, s_\mathrm{init}, s_\mathrm{goal})$.
\end{enumerate}
\end{theorem} 
Theorem \ref{theorem:planning_equivalence} (ii) says that, to look for an optimal trajectory to $({Z}, s_\mathrm{init}, s_\mathrm{goal})$, one can look for an optimal trajectory to $(\bar{Z}, s_\mathrm{init}, s_\mathrm{goal})$ instead and then just expand it. This result is crucial for our planning algorithm, since it drastically  reduces the search space. Also, (i) says that if there is no feasible trajectory to $({Z}, s_\mathrm{init}, s_\mathrm{goal})$, then there is no feasible trajectory to $(\bar{Z}, s_\mathrm{init}, s_\mathrm{goal})$ either. Thus, we can exclusively consider the reduced planning objective $(\bar{Z}, s_\mathrm{init}, s_\mathrm{goal})$. For the planning algorithm, we also need the following~result:
\begin{proposition}\label{theorem:from_B}
Consider $(\bar{Z},s_\mathrm{init},s_\mathrm{goal})$. Then:
\begin{enumerate}[(i)]
\item An optimal trajectory $z = (q_i,x_i)_{i=1}^N$ to $(\bar{Z},s_\mathrm{init},s_\mathrm{goal})$ has a state $q_i$ equal to $\mathrm{start}(\bar{B})$. 
\item Let $\mathrm{start}(\bar{B})$ be a state of $\bar{D}_{k_1}$ and $\mathrm{start}(\bar{D}_{k_1-1})$ be the start state of $\bar{D}_{k_2}$, and so on till $\mathrm{start}(\bar{D}_{k_j})$ is the start state of $\bar{D}_1$. Then, provided $(\bar{Z},\mathrm{start}(\bar{B}),s_{goal})$ is feasible\footnote{This is true if $(\bar{Z}, s_\mathrm{init}, s_\mathrm{goal})$ is feasible, by (i).}, an optimal trajectory to $(\bar{Z},\mathrm{start}(\bar{B}),s_{goal})$ is $w_1,w_2,\dots,w_j$, where $w_i$ is an optimal trajectory in $\bar{D}_{k_i}$ from the start state of $\bar{D}_{k_i}$ to the state corresponding to $\bar{D}_{k_i-1}$ (with $\bar{D}_0 = s_{goal}$).
\end{enumerate}
\end{proposition}

\begin{algorithm}[t]
\caption{Reduce\_and\_solve}\label{alg:Step1} 
\begin{algorithmic}[1]
\Require $(Z,s_\mathrm{init},s_\mathrm{goal})$ and $(c_x^M,z_x^M)_{x \in \Sigma, M \in X}$.
\Ensure An optimal trajectory $z$ to $(\bar{Z},s_\mathrm{init},s_\mathrm{goal})$.
\State \textbf{Part 1: Solve $(\bar{Z},s_\mathrm{init},\mathrm{start}(\bar{B}))$}
\State $\bar{Z} \gets \mathrm{Reduce}(Z,s_{\mathrm{init}},s_{\mathrm{goal}},(c_x^M,z_x^M)_{x \in \Sigma, M \in X})$
\State Set $G$ to an empty graph \Comment{To be constructed}
\State Add nodes to $G$ 
\For {$i =1,\dots,n$} \Comment{Consider $\bar{U}_{i}$}
\State Get states to search from and too, $I$ and $D$, in $\bar{U}_i$.
\For {$s \in I$}
\State $(c_d, z_d)_{d \in D} \gets \mathrm{Dijkstra}(s, D, \bar{U}_{i})$
\State Add arcs corresponding to $(c_d, z_d)_{d \in D}$ in $G$ 
\EndFor
\EndFor
\State $z_1 \gets \mathrm{Dijkstra}(s_\mathrm{init}, \bar{B}, G)$ \Comment{End of Part 1}
\State \textbf{Part 2: Solve $(\bar{Z},\mathrm{start}(\bar{B}),s_\mathrm{goal})$}
\State Let $z_2$ be an empty trajectory \Comment{To be constructed}
\State $(M,c) \gets (\bar{D}_{k_1},\mathrm{start}(\bar{B}))$ 
\State $g \gets$ corresponding state in $M$ to $\bar{D}_{k_1-1}$ 
\While {$c \neq s_\mathrm{goal}$}
\State $z \gets \mathrm{Dijkstra}(c,g,M)$ 
\State $z_2 \gets z_2 z$ \Comment{Add $z$ to $z_2$}
\State $c \gets start(g)$
\If {$c \neq s_\mathrm{goal}$}
\State Let $\bar{D}_{k_i}$ be the MM with state $c$
\State $M \gets \bar{D}_{k_i}$
\State $g \gets$ corresponding state in $M$ to $\bar{D}_{k_i-1}$ 
\EndIf
\EndWhile \Comment{End of Part 2}
\State return $z \gets z_1 z_2$
\end{algorithmic}
\end{algorithm}

\subsection{Online Step: Algorithm}\label{online_step_algorithm}
In this section, we provide the details of the online step, based on the theory in Section \ref{online_step}. More precisely, to compute an optimal plan to $(Z,s_\mathrm{init},s_\mathrm{goal})$, the algorithm first considers the reduced planning objective $(\bar{Z},s_\mathrm{init},s_\mathrm{goal})$, which by Proposition \ref{theorem:from_B} (i) can be divided into obtaining an optimal trajectory to $(\bar{Z},s_\mathrm{init},\mathrm{start}(\bar{B}))$ and then combining it with an optimal trajectory to $(\bar{Z},\mathrm{start}(\bar{B}),s_\mathrm{goal})$. An optimal trajectory to $(Z,s_\mathrm{init},s_\mathrm{goal})$ is then obtained by expanding this trajectory, as given by Theorem \ref{theorem:planning_equivalence} (ii). From this, we obtain an optimal plan to $(Z,s_\mathrm{init},s_\mathrm{goal})$. The details are given below.


\subsubsection{Solving \texorpdfstring{$(\bar{Z},s_\mathrm{init},s_\mathrm{goal})$}{(barZ,sinit,sgoal)}}

To solve $(\bar{Z},s_\mathrm{init},s_\mathrm{goal})$, with procedure given by Algorithm \ref{alg:Step1}, we first solve $(\bar{Z},s_\mathrm{init},\mathrm{start}(\bar{B}))$. To this end, we first reduce $Z$ to $\bar{Z}$. We then note that there exists an optimal trajectory from $s_\mathrm{init}$ to $\mathrm{start}(\bar{B})$ in $\bar{Z}$ that only goes through states in $\bar{U}_1,\dots,\bar{U}_n$. Therefore, to solve $(\bar{Z},s_\mathrm{init},\mathrm{start}(\bar{B}))$, we only need to consider $\bar{U}_1,\dots,\bar{U}_n$. Furthermore, only some states and trajectories in each $\bar{U}_i$ are relevant. Namely, for $\bar{U}_1$, the relevant states are the ones we might start from: $s_\mathrm{init}$ and $s(\bar{U}_1)$. From these states, the relevant trajectories are the ones that optimally exit $\bar{U}_1$, starting from the states.\footnote{More precisely, if one can exit $\bar{U}_1$ with $x \in \Sigma$, starting from $s_\mathrm{init}$, then we find a trajectory $t_x$ that does this optimally. The relevant trajectories from $s_\mathrm{init}$ are then all the found $t_x$. The case for $s(\bar{U}_1)$ is analogous.} 
For $\bar{U}_2$, the relevant states are $s(\bar{U}_2)$ and states in $\bar{U}_2$ that could be reached by exiting $\bar{U}_1$ (at most $|\Sigma|$ such states). From these states, the relevant trajectories are the ones that optimally exit $\bar{U}_2$ as well as the optimal trajectories to get to $\bar{U}_1$ (might be optimal to go back to $\bar{U}_1$). The other $\bar{U}_i$ are analogous to $\bar{U}_2$ except that: the relevant trajectories of $\bar{U}_\alpha$ also include the optimal trajectories to $\bar{B}$; and, for $\bar{U}_n$, we do not calculate the ones that optimally exit $\bar{U}_n$ (since this would only stop $\bar{Z}$). With this, we form a graph $G$ where nodes (arcs) corresponds to the relevant states (trajectories), labelling each arc with the relevant trajectory and its cumulative cost. Searching in $G$ using Dijkstra's algorithm, from $s_\mathrm{init}$ to $\bar{B}$, we find an optimal trajectory $z_1$ to $(\bar{Z},s_\mathrm{init},\mathrm{start}(\bar{B}))$. This procedure is Part 1 in Algorithm \ref{alg:Step1}.
We then get an optimal trajectory $z_2$ to $(\bar{Z},\mathrm{start}(\bar{B}),s_\mathrm{goal})$ using Proposition \ref{theorem:from_B} (ii), and Dijkstra's algorithm to search in each $\bar{D}_{k_i}$. This is Part 2 in Algorithm \ref{alg:Step1}. Finally, we combine $z_1$ and $z_2$ to get an optimal trajectory $z$ to $(\bar{Z},s_\mathrm{init},s_\mathrm{goal})$. We get time complexity:

\begin{proposition}\label{prop:time_complexity_step_1_and_2}
The time complexity of Algorithm \ref{alg:Step1} is
\begin{align*}
O \big (|\Sigma|^2 \mathrm{depth}(Z) + |\Sigma| \mathrm{depth}(Z) \cdot \log ( |\Sigma| \mathrm{depth}(Z) ) \big ) 
+ \\
O \big ( [b_s |\Sigma | + b_s \log(b_s)] \cdot \mathrm{depth}(Z) \big )
\end{align*}
where the first (second) $O$-term is from Part 1 (Part 2), and $b_s$ is the maximum number of states in an MM. In particular, with bounded $|\Sigma|$ and $b_s$, we get $O(\mathrm{depth}(Z) \cdot \log (\mathrm{depth}(Z)))$.
\end{proposition}

\subsubsection{Solving \texorpdfstring{$(Z,s_\mathrm{init},s_\mathrm{goal})$}{(Z,sinit,sgoal)}}
We get an optimal plan $u$ to $(Z,s_\mathrm{init},s_\mathrm{goal})$ by conducting an optimal expansion of the optimal trajectory $z$ to $(\bar{Z},s_\mathrm{init},s_\mathrm{goal})$ from Algorithm \ref{alg:Step1}. The optimal plan $u$ can be executed sequentially or obtained at once, with procedure given by Algorithm \ref{alg:step_3}. Here, $\mathrm{Plan\_expansion}$ is identical to $\mathrm{Trajectory\_expansion}$ in Algorithm \ref{alg:optimal_expansion} except line 4 that is changed to [return $x$] or [apply $x$] if one wants the full plan $u$ at once or sequential execution, respectively. We get time complexity:

\begin{proposition}\label{prop:time_complexity_step_3}
Executing an optimal plan $u$ sequentially using Algorithm \ref{alg:step_3} has time complexity $O(\mathrm{depth}(Z))$ to obtain the next input in $u$. Obtaining the full optimal plan at once has time complexity $O(\mathrm{depth}(Z) |u|)$, where $|u|$ is the length of $u$.
\end{proposition}

\begin{algorithm}[t]
\caption{Expand}\label{alg:step_3}
\begin{algorithmic}[1]
\Require Optimal trajectory $z$ to $(\bar{Z},s_\mathrm{init},s_\mathrm{goal})$ and $(c_x^M,z_x^M)_{x \in \Sigma, M \in X}$.
\Ensure Executes/saves optimal plan $u$ to $(Z,s_\mathrm{init},s_\mathrm{goal})$.
\State Let $u$ be an empty trajectory \Comment{To be constructed} 
\For {$(q,x)$ in $z$}
\State $t \gets \mathrm{Plan\_expansion}(q,x)$ 
\State $u \gets u t$ \Comment{Concatenate $u$ and $t$}
\EndFor
\If {save $u$}
return $u$
\EndIf
\end{algorithmic}
\end{algorithm}

\section{Numerical evaluations}\label{numerical_evaluations}
In this section, we consider numerical case studies to validate the hierarchical planning algorithm given by Algorithm~\ref{alg:hierarchical_planning}. Case study 1 demonstrates the scalability of the algorithm. Case study 2 show-case the algorithm on the robot application introduced in the motivation.


\subsection{Case Study 1: Recursive System}\label{recursive_example}

\subsubsection{Setup}
To validate the efficiency of Algorithm~\ref{alg:hierarchical_planning}, we consider an HiMM $Z$ constructed recursively by nesting the same MM repeatedly to a certain depth. More precisely, the MM $M$ we consider has tree states $Q = \{1,2,3\}$ with start state 2 and three inputs $\Sigma = \{x,y,z\}$ with transitions depicted in Fig. \ref{fig:recursive_example} (left) and units costs (not depicted). The recursion step is then done by replacing state 1 and 3 with $M$, with result given by Fig. \ref{fig:recursive_example} (right). This procedure is then repeated for the recently added MMs until we have reached a certain depth (e.g., in Fig. \ref{fig:recursive_example}, we would replace state 2, 4, 5 and 7 with $M$). This yields the HiMM $Z$.

We check the computing time of Algorithm~\ref{alg:hierarchical_planning} when varying the depth of $Z$. To this end, we let $s_\mathrm{init}$ and $s_\mathrm{goal}$ be opposite states in the HiMM (e.g., state 2 and 7 in Fig. \ref{fig:recursive_example}) and compute the full optimal plan. For comparison, we also consider Dijkstra's algorithm for finding the optimal plan, applied to the equivalent flat MM.\footnote{This flat MM can be obtained by simply checking all transitions and costs at each state in $Z$, and form the corresponding MM from this data.} 


\begin{figure}[t]
\centering
\includegraphics[width=0.7\linewidth]{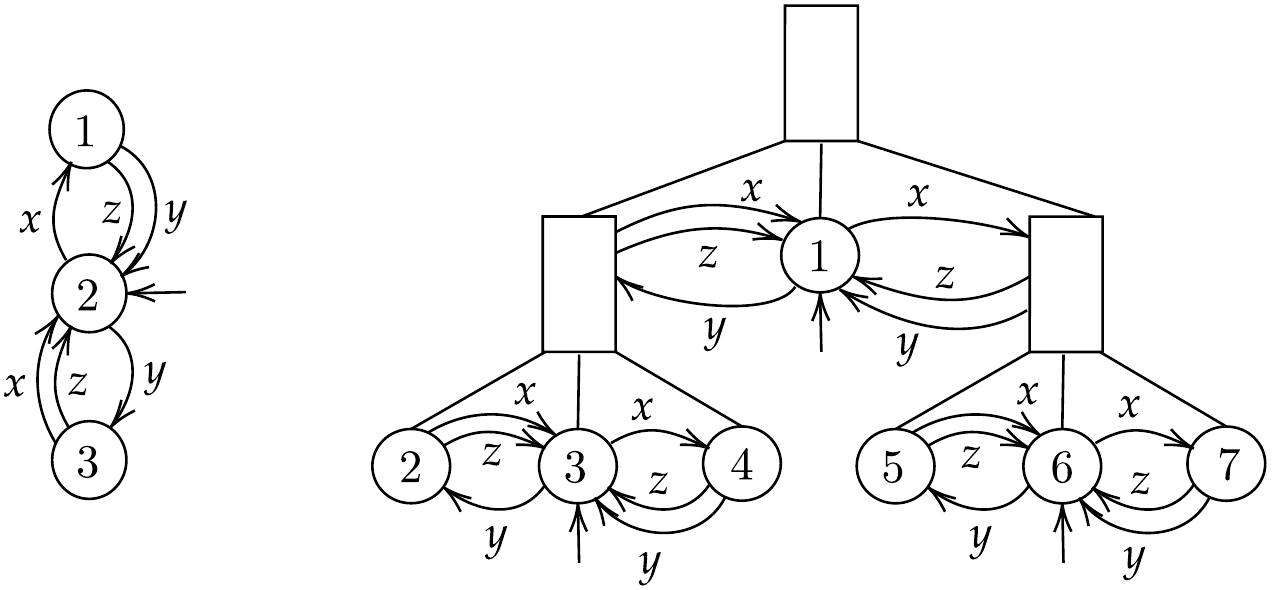}
\caption{MM $M$ (left) and HiMM $Z$ (right) with depth 2 for Case study 1.}
\label{fig:recursive_example}
\end{figure}  

\subsubsection{Result}

The computing time is shown in Fig. \ref{fig:recursive_example_result} with $\mathrm{depth}(Z)$ from 1 to 20. 
The online step finds optimal plans for all depths within milliseconds, while Dijkstra's algorithm is a bit faster for small depths ($\mathrm{depth}(Z) \leq 6$) but takes several seconds for larger depths ($\mathrm{depth}(Z) \geq 15$). In particular, for $\mathrm{depth}(Z)=20$, with about 2 million states, the online step finds an optimal plan in just 3.8 ms, compared to 108 s using Dijkstra's algorithm. Also, the offline step has a computing time comparable to Dijkstra's algorithm being slower for small depths (max 4 times slower), but negligible difference for large depths, and even slightly faster for $\mathrm{depth}(Z) \geq 18$ (due to a slower increase). Hence, for large depth, even the computing time for the offline plus the online step is slightly faster than Dijkstra's~algorithm.

\subsection{Case Study 2: Robot Warehouse Application}\label{robot_example}

\subsubsection{Setup}
We now consider the robot application introduced in the motivation, schematically depicted by Fig. \ref{fig:motivating_example}. To formalise the example as an HiMM $Z$, we consider 10 warehouses ordered linearly, where the robot can move to any neighbouring house (e.g., to house 2 and 4 from house 3, and only to house 2 from house 1), at a cost of 100. This yields the MM $M_1$ corresponding to the top layer of $Z$, with house 1 as start state and inputs $\Sigma = \{\text{left}, \text{right}\}$. 

Furthermore, each house is modelled as an MM $M_2$ having a single room. The room is a square $10 \times 10$ grid, where the robot at grid point $(i,j)$ ($1 \leq i,j \leq 10$) can move to any neighbouring (non-diagonal) grid point at a cost of 1 using inputs $\Sigma = \{\text{left}, \text{right}, \text{up}, \text{down}\}$. We also have a grid-point just outside the house, the entrance state, adjacent to $(1,1)$. The robot can move between the entrance state and $(1,1)$ with cost 1. From the entrance state, the robot can also exit $M_2$ by applying input $a \in \{\text{left}, \text{right}\}$. This $a$ is then fed to $M_1$, which moves the robot to the corresponding neighbouring house. The MM $M_2$ has $101$ states and four actions with the entrance state as start state. 

\begin{figure}[t]
\centering
\includegraphics[width=0.9\linewidth]{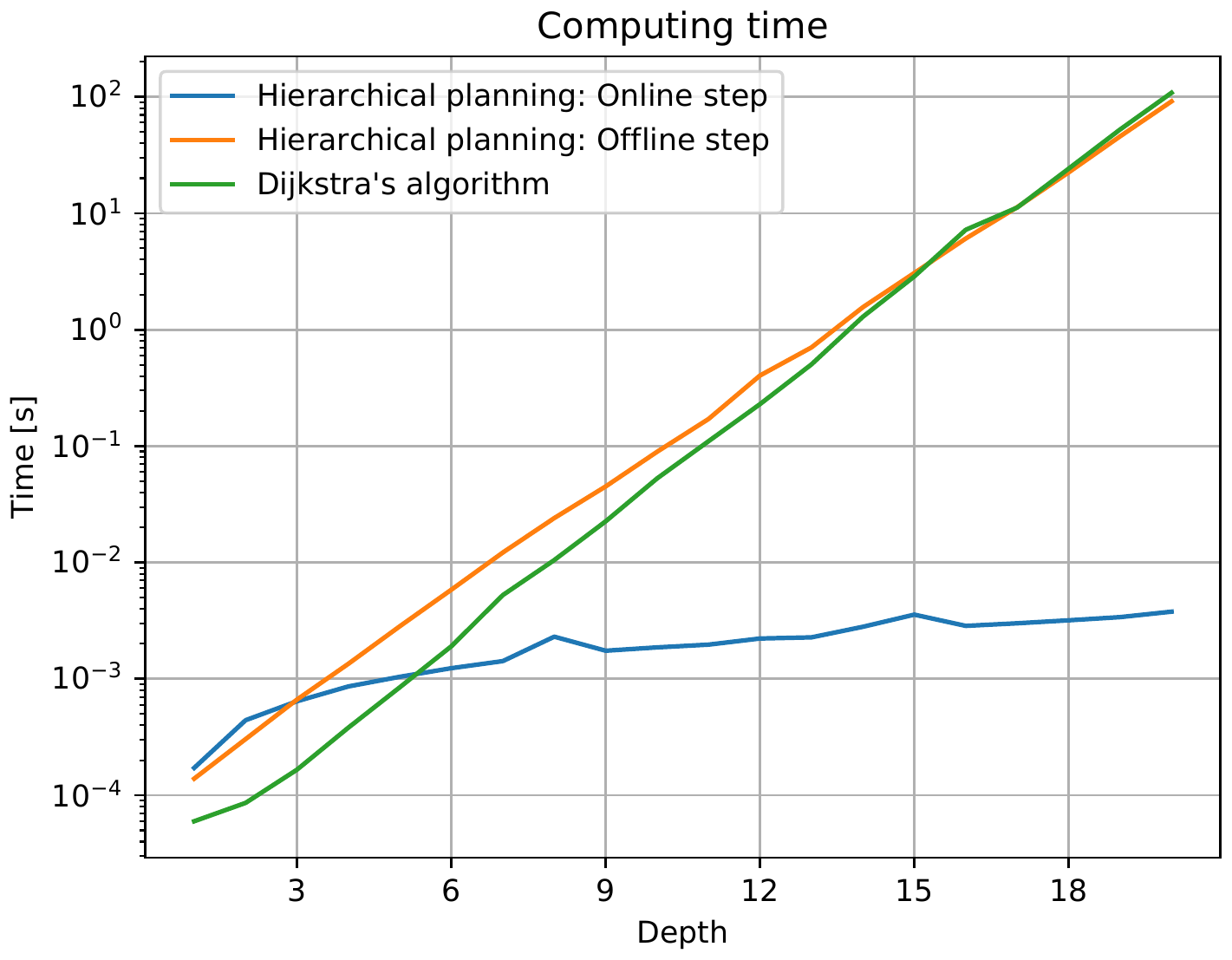}
\caption{Computing time for varying depth in Case Study 1.}
\label{fig:recursive_example_result}
\end{figure}  

Finally, at each grid-point inside the house, the robot has a work desk, modelled as an MM $M_3$, consisting of 9 laboratory test tubes arranged in a $3 \times 3$ test tube rack, where the robot can move between the test tubes or scan a tube using a robot arm. More precisely, at test tube $(i,j)$ ($1 \leq i,j \leq 3$), it can move the robot arm to any neighbouring tube (analogous to $M_2$) or scan the test tube, using inputs $\Sigma = \{\text{left}, \text{right}, \text{up}, \text{down}\} \cup \{ \text{scan} \}$. When a test tube has been scanned, it remembers it and do not scan other tubes. Similar to $M_2$, we have an entrance state from which the robot can either enter the work desk (starting at test tube $(1,1)$ and nothing scanned), or exit $M_3$ by applying input $a \in \{\text{left}, \text{right}, \text{up}, \text{down}\}$. This input $a$ is then fed to $M_2$ which transition to the corresponding grid-point. From $(1,1)$, we can also go back to the entrance state. All transition costs are set to 0.5 except the scanning, which costs 10. The MM $M_3$ has $9 \cdot 10 +1 = 91$ states and inputs $\Sigma$, with entrance state as start state. The hierarchy of $M_1$, $M_2$ and $M_3$ yields~$Z$.



\subsubsection{Result}
We set $s_\mathrm{init}$ to be the state where the robot is in house 1 at grid-point $(10,10)$ having scanned test tube $(3,3)$, and $s_\mathrm{goal}$ is identical to $s_\mathrm{init}$ except in house 10. That is, the robot has to move to house 10 and scan test tube $(3,3)$ at grid-point $(10,10)$. The online step finds an optimal plan in just 0.022 s compared to 3.4 s using Dijkstra's algorithm. Also, the offline step takes only 2.0 s, hence, even the offline plus online step is faster than Dijkstra's algorithm.

\section{Conclusion}\label{conclusion}
In this paper, we have considered a planning problem for an HiMM and developed an algorithm for efficiently computing optimal plans between any two states. The algorithm consists of an offline step and an online step. The offline step computes exit costs for each MM in a given HiMM $Z$. This step is done only once for $Z$, with time complexity scaling linearly with the number of MMs in $Z$. The online step then computes an optimal plan, from a given initial state to a goal state, by constructing an equivalent reduced HiMM $\bar{Z}$ (based on the exit costs), computing an optimal trajectory $z$ for $\bar{Z}$, and finally expanding $z$ to obtain an optimal plan $u$ to $Z$. The online step finds an optimal trajectory $z$ to $\bar{Z}$ in time $O(\mathrm{depth}(Z) \log ( \mathrm{depth}(Z) ) )$ and obtains the next optimal input in $u$ from $z$ in time $O(\mathrm{depth}(Z) )$, or the full optimal plan $u$ in time $O(\mathrm{depth}(Z) |u|)$. We validated our algorithm on large HiMMs having up to 2 million states, including a mobile robot application, and compared our algorithm with Dijkstra's algorithm. Our algorithm outperforms the latter, where the partition into an offline and online step reduces the overall computing time for large systems, and the online step computes optimal plans in just milliseconds compared to tens of seconds using Dijkstra's~algorithm. 




Future work includes extending the algorithm to efficiently handle changes in the given HiMM (e.g., modifying an MM in the HiMM), and comparing it with other hierarchical methods. Another challenge is to extend the setup to stochastic~systems.


\bibliographystyle{plain}
\bibliography{Ref3}

\begin{thebibliography}{10}

\bibitem{alur1998model}
R.~Alur and M.~Yannakakis.
\newblock Model checking of hierarchical state machines.
\newblock {\em ACM SIGSOFT Software Engineering Notes}, 23(6):175--188, 1998.

\bibitem{bast2016route}
H.~Bast, et~al.
\newblock Route planning in transportation networks.
\newblock In {\em Algorithm engineering}, pages 19--80. Springer, 2016.

\bibitem{biggar2021modular}
O.~Biggar, M.~Zamani, and I.~Shames.
\newblock Modular decomposition of hierarchical finite state machines.
\newblock {\em arXiv preprint arXiv:2111.04902}, 2021.

\bibitem{cassandras2008introduction}
C.~G. Cassandras and S.~Lafortune.
\newblock {\em Introduction to Discrete Event Systems}.
\newblock Springer, 2008.

\bibitem{dibbelt2016customizable}
J.~Dibbelt, B.~Strasser, and D.~Wagner.
\newblock Customizable contraction hierarchies.
\newblock {\em Journal of Experimental Algorithmics (JEA)}, 21:1--49, 2016.

\bibitem{Dijkstra1959}
E.~W. Dijkstra.
\newblock A note on two problems in connexion with graphs.
\newblock {\em Numer. Math.}, 1(1):269–271, dec 1959.

\bibitem{DijkstraFibonacci}
M.L. Fredman and R.E. Tarjan.
\newblock Fibonacci heaps and their uses in improved network optimization
  algorithms.
\newblock In {\em 25th Annual Symposium on Foundations of Computer Science,
  1984.}, pages 338--346, 1984.

\bibitem{harel1987statecharts}
D.~Harel.
\newblock Statecharts: A visual formalism for complex systems.
\newblock {\em Science of computer programming}, 8(3):231--274, 1987.

\bibitem{lavalle2006planning}
S.~M. LaValle.
\newblock {\em Planning algorithms}.
\newblock Cambridge university press, 2006.

\bibitem{ma2006nonblocking}
C.~Ma and W.~M. Wonham.
\newblock Nonblocking supervisory control of state tree structures.
\newblock {\em IEEE Transactions on Automatic Control}, 51(5):782--793, 2006.

\bibitem{10.1145/1498698.1564502}
J.~Maue, P.~Sanders, and D.~Matijevic.
\newblock Goal-directed shortest-path queries using precomputed cluster
  distances.
\newblock {\em ACM J. Exp. Algorithmics}, 14, 2010.

\bibitem{Mealy1955}
George~H. Mealy.
\newblock A method for synthesizing sequential circuits.
\newblock {\em The Bell System Technical Journal}, 34(5):1045--1079, 1955.

\bibitem{millington2018artificial}
I.~Millington and J.~Funge.
\newblock {\em Artificial intelligence for games}.
\newblock CRC Press, 2018.

\bibitem{mohring2007partitioning}
R.~H. M{\"o}hring, et~al.
\newblock Partitioning graphs to speedup dijkstra's algorithm.
\newblock {\em Journal of Experimental Algorithmics (JEA)}, 11:2--8, 2007.

\bibitem{schillinger2016human}
P.~Schillinger, S.~Kohlbrecher, and O.~Von~Stryk.
\newblock Human-robot collaborative high-level control with application to
  rescue robotics.
\newblock In {\em 2016 IEEE Int. Conf. Robot. Autom. (ICRA)}, pages 2796--2802.
  IEEE, 2016.

\bibitem{timo2014reachability}
O.~N Timo, et~al.
\newblock Reachability in hierarchical machines.
\newblock In {\em Proceedings of the 2014 IEEE 15th International Conference on
  Information Reuse and Integration (IEEE IRI 2014)}, pages 475--482. IEEE,
  2014.

\bibitem{wang2020real}
X.~Wang, Z.~Li, and W.~M. Wonham.
\newblock Real-time scheduling based on nonblocking supervisory control of
  state-tree structures.
\newblock {\em IEEE Transactions on Automatic Control}, 66(9):4230--4237, 2020.

\bibitem{yannakakis2000hierarchical}
M.~Yannakakis.
\newblock Hierarchical state machines.
\newblock In {\em IFIP International Conference on Theoretical Computer
  Science}, pages 315--330. Springer, 2000.

\end{thebibliography}

\if\longversion1

\section*{Appendix}

\subsection{Proof of Proposition \ref{th:offline_optimal_cost}}

\begin{proof}[Proof of Proposition \ref{th:offline_optimal_cost}]
We prove this by induction over the tree $T$, where $c_x^M$ are the values computed by Algorithm~\ref{alg:optimal_exit_table}. 


For the base case, consider an MM $M$ of $Z$ such that $M$ has no MMs as children in $T$. In this case, any state $q$ of $M$ is a state of $Z$, hence all optimal $(q,x)$-exit costs are zero, agreeing with Algorithm~\ref{alg:optimal_exit_table} setting $c_x^q = 0$.

For the induction step, consider any MM $M$ of $Z$. For each $q \in Q(M)$, assume that $c_x^q$ (computed by Algorithm~\ref{alg:optimal_exit_table}) is the optimal $(q,x)$-exit cost.\footnote{In particular, note that this assumption is true for the base case.} We will show that $c_x^M$ (computed by Algorithm~\ref{alg:optimal_exit_table}) is then the optimal $(M,x)$-exit cost. The proof then follows by induction over $T$ (starting with MMs as in the base case, and go higher up in $T$ using the induction step).




To show that $c_x^M$ is the optimal $(M,x)$-exit cost, note first that if no $(M,x)$-exit trajectory exists, then $c_x^M = \infty$. To see this, assume by contradiction that $c_x^M < \infty$ and let $z_x^M = \{(s_1,x_1)\}_{i=1}^N$ be the corresponding trajectory in $M$. Each $c_{x_i}^{s_i}$ must be finite and, by assumption, equal to the optimal $(s_i,x_i)$-cost. Let $t_i$ be the corresponding optimal $(s_i,x_i)$-exit trajectory. Then, $t_1 t_2 \dots t_N$ is a $(M,x)$-exit trajectory, a contradiction. Hence, $c_x^M = \infty$. 

We are left with the case when a $(M,x)$-exit trajectory exists. Let $t$ be an optimal $(M,x)$-exit trajectory (by finiteness, such a trajectory must exit). Note that $t$ can be partitioned into $t = t_1 t_2 \dots t_N$, where each $t_i$ is a $(s_i,x_i)$-exit trajectory for some $s_i \in Q(M)$. By optimality, each $t_i$ results in a cumulated cost $c_{x_i}^{s_i}$ plus transition cost $\gamma(s_i,x_i)$, except $t_N$ that only results in exit cost $c_{x_N}^{s_N}$ (since no transition from $s_N$ for $x_N=x$ exists by assumption). Hence, the trajectory $z = \{(s_1,x_1)\}_{i=1}^N$ in $\hat{M}$ has the same cumulated cost as $t$. Assume by contradiction that there exists a trajectory $\tilde{z} = \{(\tilde{s}_i,\tilde{x}_i)\}_{i=1}^N$ of $\hat{M}$ that reaches $E_x$ with a lower cost than $z$. Consider $\tilde{t} = \tilde{t}_1 \dots \tilde{t}_{\tilde{N}}$ where $\tilde{t}_i$ is the corresponding optimal $(\tilde{s}_i,\tilde{x}_i)$-exit trajectory. Then, $\tilde{t}$ has the same cost as $\tilde{z}$, which is lower than the cumulated cost of $z$, thus also of $t$. This is a contradiction. We conclude that $t$ must have the same cumulated cost as an optimal trajectory in $\hat{M}$ reaching $E_x$, and therefore, $c_x^M$ equals the optimal $(M,x)$-exit cost. By induction, this concludes the~proof.
\end{proof}

\subsection{Proof of Proposition \ref{th:offline_time_complexity}}
\begin{proof}[Proof of Proposition \ref{th:offline_time_complexity}]
We construct $\hat{M}$ from $M$ by adding the states $\{E_x\}_{x \in \Sigma}$ in time $O(|\Sigma |)$, constructing $\hat{\delta}$ by going through all the values of the function in time $O(b_s | \Sigma |)$, and constructing  $\hat{\gamma}$ analogously. The total time complexity for line 11 in Algorithm \ref{alg:optimal_exit_table} is therefore $O(b_s | \Sigma |)$. Furthermore, note that the maximum number of states of $\hat{M}$ is $b_s+|\Sigma|$ and the maximum number of transition arcs of $\hat{M}$ is $b_s |\Sigma|$ (considering $\hat{M}$ as a graph). Searching with Dijkstra's algorithm \cite{DijkstraFibonacci}, line 12 in Algorithm \ref{alg:optimal_exit_table}, has therefore complexity $O(E+V \log V) = O(b_s |\Sigma|+(b_s+|\Sigma|) \log(b_s+|\Sigma|))$ (where $E$ is the number of arcs and $V$ is the number of (graph) nodes in the graph used in Dijkstra's algorithm). Finally, line 3-10 takes time $O(b_s |\Sigma|)$ excluding the time it takes to compute $\mathrm{Optimal\_exit}(M_q)$ since that time is already accounted for (when considering the MM $M_q$). We conclude that the total time spent on one MM $M$ is $O(b_s |\Sigma|+(b_s+|\Sigma|) \log(b_s+|\Sigma|))$. This is done for all MMs in $Z$, so assuming there are $N$ MMs in $Z$, the time complexity for Algorithm \ref{alg:optimal_exit_table} is
\begin{align}\label{eq:time_complexity_1}
O(N [b_s |\Sigma|+(b_s+|\Sigma|) \log(b_s+|\Sigma|)]).
\end{align}
This completes the proof.
\end{proof}

\subsection{Proof of Proposition \ref{proposition:optimal_expansion}}

\begin{proof}[Proof of Proposition \ref{proposition:optimal_expansion}]
We may assume that there exists a $(q,x)$-exit trajectory (the result is otherwise trivial). We prove that $E(q,x)$ is an optimal $(q,x)$-exit trajectory by induction over the depth of the subtree of $T$ with root $q$. The base case when $q$ is a state of $Z$ is clear. Proceed by induction. Assume $q$ is not a state with corresponding MM $M$. Then $E(q,x)$ is on the form 
\begin{equation*}
E(q,x) = z_1 z_2 \dots z_m
\end{equation*}
as given by Algorithm \ref{alg:optimal_expansion}. By assumption, each $z_i$ is an optimal $(q_i,x_i)$-exit trajectory and hence the total exit cost is $\sum_{i=1}^{m-1} [c^{q_i}_{x_i}+\gamma_M(q_i,x_i)]+c^{q_m}_{x_m}$. Assume by contradiction that $t$ is an optimal $(q,x)$-exit trajectory with lower $(q,x)$-exit cost than $E(q,x)$. Note that, $t$ must be on the form 
$t = t_1 t_2 \dots t_l$,
where each $t_i$ is a $(s_i,w_i)$-exit trajectory, for some state $s_i$ of $M$ and input $w_i$. By optimality, each $t_i$ must be an optimal $(s_i,w_i)$-exit trajectory, and hence the $(q,x)$-exit cost of $t$ is $\sum_{i=1}^{l-1} [c^{s_i}_{x_i}+\gamma_M(s_i,x_i)]+c^{s_l}_{x_l}$. However, $\{(q_i,x_i)\}_{i=1}^m$ is an optimal trajectory to $E_x$ in $\hat{M}$, so its cost must be lower or equal to the cost of $\{(s_i,w_i)\}_{i=1}^l$. That is,
\begin{equation*}
\sum_{i=1}^{m-1} [c^{q_i}_{x_i}+\gamma_M(q_i,x_i)]+c^{q_m}_{x_m} \leq \sum_{i=1}^{l-1} [c^{s_i}_{x_i}+\gamma_M(s_i,x_i)]+c^{s_l}_{x_l},
\end{equation*}
which contradicts the assumption. We conclude that $E(q,x)$ is an optimal $(q,x)$-exit trajectory. 
\end{proof}

\subsection{Proof of Theorem \ref{theorem:planning_equivalence}}

To prove Theorem \ref{theorem:planning_equivalence}, we need the following two lemmas.

\begin{lemma}\label{theorem:equavalence_lemma1}
Let $z$ be an optimal trajectory to $(Z,s_\mathrm{init}, s_\mathrm{goal})$. Then $z$ can be partition into subsequences
\begin{equation}\label{eq:z_partitioned}
z = t_1 \dots t_L,
\end{equation}
where each $t_i$ is an optimal $(q_i,x_i)$-exit trajectory (in $Z$) for some node $q_i$ of $Z$ such that $q_i$ is a state of $\bar{Z}$.
Furthermore, the reduced trajectory of $z$ with respect to $\M$ equals $\bar{z} = \{(q_i,x_i)\}_{i=1}^L$ and is a trajectory to $\bar{Z}$ such that $q_1 = s_\mathrm{init}$, $q_{L+1} = \bar{\chi}(q_L,x_L) = s_\mathrm{goal}$ and $\bar{C}(\bar{z}) = C(z)$.
\end{lemma}
\begin{proof}
The partition is proved by induction as follows. Note that $z$ starts in $s_\mathrm{init}$, so $t_1 = (s_\mathrm{init}, x_1)$. By induction, assume that we have concluded that $z$ is on the form
\begin{equation*}
z = t_1 \dots t_l z_l
\end{equation*}
with some remaining trajectory $z_l$, where each $t_i$ is an optimal $(q_i,x_i)$-exit trajectory. Then $t_1 \dots t_l$ is such that it will exit $q_l$ with $x_l$ where $q_l$ is a state of $\bar{Z}$. Therefore, the first state of $z_l$ will be a start state, $\mathrm{start}(M_{l+1})$, for some MM $M_{l+1}$ in $Z$, where we may take $M_{l+1}$ such that the corresponding node $q_{l+1}$ of $M_{l+1}$ gets reduced to a state in $\bar{Z}$, or is already a state of $Z$. Note that, due to optimality, the only case when $z$ does not exit $q_{l+1}$ is the case when $q_{l+1}=s_\mathrm{goal}$. In this case, $z_l = \emptyset$. If this is not the case, then we must, due to optimality, eventually exit $q_{l+1}$. Therefore, $z$ is on the form $z = t_1 \dots t_l t_{l+1} z_{l+1}$ for some trajectory $t_{l+1}$ that, by optimality, must be an optimal $(q_{l+1},x_{l+1})$-exit trajectory. By induction, $z$ is on the form \eqref{eq:z_partitioned}. To derive the form of $\bar{z}$, simply note that $t_i$ gets reduced to $(q_i,x_i)$ and hence $\bar{z} = \{(q_i,x_i)\}_{i=1}^L$. Furthermore, the last state-input pair of $t_L$ transits to $s_\mathrm{goal}$ and $\bar{\chi}(q_L,x_L) =  \chi(q_L,x_L) = s_\mathrm{goal}$. Finally, to derive $\bar{C}(\bar{z}) = C(z)$, note first that the cost of $t_i$  is 
\begin{equation*}
c^{q_i}_{x_i} + \chi(q_i,x_i) =  \bar{\chi}(q_i,x_i)
\end{equation*}
due to optimality and the construction of $\bar{Z}$. This cost coincides with the cost of $(q_i,x_i)$ in $\bar{Z}$ and hence $\bar{C}(\bar{z}) = C(z)$.
\end{proof}

\begin{lemma}\label{theorem:equavalence_lemma2}
Let $z = (q_i,x_i)_{i=1}^L$ be a feasible trajectory of $(\bar{Z},s_\mathrm{init}, s_\mathrm{goal})$.
Then, the optimal expansion $z|^{\M} = (s_i,a_i)_{i=1}^N$ of $z$ is a trajectory of $Z$ such that $s_1 = s_\mathrm{init}$, $s_{N+1} = {\chi}(s_N,a_N) = s_\mathrm{goal}$ and $\bar{C}(z) = C(z|^{\M})$. 
\end{lemma}

\begin{proof}
Let $M_i$ be the MM corresponding to the node $q_i$ in $Z$, applicable whenever $q_i$ is not a state of $Z$. We first note that 
\begin{equation*}
z|^{\M} = E(q_1,x_1) E(q_2,x_2) \dots E(q_L,x_L)
\end{equation*}
is a trajectory. Indeed, $E(q_i,x_i)$ is a trajectory that goes from $\mathrm{start}(M_i)$ (or $q_i$ if $q_i$ is a state of $Z$) and exits $q_i$ with $x_i$. In $\bar{Z}$, $\bar{\psi}(q_i,x_i) = q_{i+1}$, so ${\psi}(q_i,x_i) = \mathrm{start}(M_{i+1})$ (or ${\psi}(q_i,x_i) = q_{i+1}$ if $q_{i+1}$ is a state of $Z$), which is the start state of $E(q_{i+1},x_{i+1})$, so $z|^{\M}$ is a trajectory. Furthermore, by Proposition \ref{proposition:optimal_expansion}, $E(q,x)$ is an optimal $(q_i,x_i)$-exit trajectory. Therefore, the cumulated cost for $E(q_i,x_i)$ is $c^{q_i}_{x_i}+\chi(q_i,x_i) = \bar{\chi}(q_i,x_i)$ and hence
\begin{equation*}
\bar{C}(z) =  \sum_{i=1}^L \bar{\chi}(q_i,x_i) = C(z|^{\M}).
\end{equation*}
Note also that $E(q_1,x_1) = (s_\mathrm{init}, x_1)$, so $s_1 = s_\mathrm{init}$. Finally, $(s_N,a_N) =(s_N, x_L)$ exits $q_L$ with $x_L$ and goes to $s_\mathrm{goal}$.
\end{proof}

\begin{proof}[Proof of Theorem 1]
We first prove (i). Let $z$ be an optimal trajectory to $({Z}, s_\mathrm{init}, s_\mathrm{goal})$. By Lemma \ref{theorem:equavalence_lemma1}, the reduced trajectory
$\bar{z} = (q_i,x_i)_{i=1}^L$ is a feasible trajectory to $(\bar{Z}, s_\mathrm{init}, s_\mathrm{goal})$ such that $\bar{C}(\bar{z}) = C(z)$. Assume by contradiction that $\bar{z}$ is not an optimal trajectory. Then, 
there exists an optimal trajectory $z^*$ to $(\bar{Z}, s_\mathrm{init}, s_\mathrm{goal})$ such that $\bar{C}(z^*) < \bar{C}(\bar{z})$. By Lemma \ref{theorem:equavalence_lemma2}, $C(z^*|^{\M}) = \bar{C}(z^*) < \bar{C}(\bar{z}) = C(z)$, a contradiction, and (i) follows.

We continue with (ii). Let $z$ be an optimal trajectory to $(\bar{Z}, s_\mathrm{init}, s_\mathrm{goal})$. By Lemma \ref{theorem:equavalence_lemma2}, $z|^{\M}$ is a feasible trajectory to $(Z, s_\mathrm{init}, s_\mathrm{goal})$ such that $\bar{C}(z) = C(z|^{\M})$. Assume by contradiction that $z|^{\M}$ is not an optimal trajectory. Then, there exists an optimal trajectory $z^*$ to $(Z, s_\mathrm{init}, s_\mathrm{goal})$ such that ${C}(z^*) < {C}(z|^{\M})$. By Lemma \ref{theorem:equavalence_lemma1}, $\bar{C}(\bar{z}^*) = {C}(z^*) < {C}(z|^{\M}) = \bar{C}(z)$, a contradiction, and (ii) follows.
\end{proof}

\subsection{Proof of Proposition \ref{theorem:from_B}}

\begin{proof}
Claim (i) is clear. We prove (ii). Consider any plan $t$ from $\mathrm{start}(\bar{B})$ to $s_\mathrm{goal}$. Note that $t$ must eventually go from $\mathrm{start}(\bar{D}_{k_i})$ to the state corresponding to $\bar{D}_{k_i-1}$, and this partial cost is greater than or equal to the cumulated cost of $w_i$. Since we have no negative costs, we conclude that the cumulative cost of $t$ must be greater than or equal to the cumulative cost of  $w_1,w_2,\dots,w_j$. Hence, $w_1,w_2,\dots,w_j$ is an optimal trajectory to $(\bar{Z},\mathrm{start}(\bar{B}),s_\mathrm{goal})$.
\end{proof}

\subsection{Proof of Proposition \ref{prop:time_complexity_step_1_and_2}}

To prove Proposition \ref{prop:time_complexity_step_1_and_2}, we first provide a more detailed version of Algorithm \ref{alg:Step1} given by Algorithm \ref{alg:Step1_detailed}. Algorithm \ref{alg:Step1_detailed} calls several functions, given by Algorithm \ref{alg:compute_alpha_and_beta} to Algorithm~\ref{alg:add_arcs}.

\begin{algorithm}[h]
\caption{Reduce\_and\_solve: Detailed version}\label{alg:Step1_detailed}
\begin{algorithmic}[1]
\Require $(Z,s_\mathrm{init},s_\mathrm{goal})$ and $(c_x^M,z_x^M)_{x \in \Sigma, M \in X}$.
\Ensure An optimal trajectory $z$ to $(\bar{Z},s_\mathrm{init},s_\mathrm{goal})$.
\State \textbf{Part 1: Solve $(\bar{Z},s_\mathrm{init},\mathrm{start}(\bar{B}))$}
\State $(\bar{Z},U_{path},D_{path},\alpha,\beta)$ $\gets \mathrm{Reduce}(Z,s_{\mathrm{init}},s_{\mathrm{goal}},(c_x^M,z_x^M)_{x \in \Sigma, M \in X})$
\State $\tilde{\psi} \gets$ Compute\_transitions($\bar{Z}, U_{path}$)
\State Set $G$ to an empty graph. \Comment{To be constructed}
\State $G \gets$ Add\_nodes($\bar{Z},U_{path},D_{path},\beta,G$) \Comment{Add nodes to $G$}
\For {$i =1,\dots,n$} \Comment{Consider $\bar{U}_{i}$}
\State $(I,D) \gets$ Get\_search\_states($\bar{Z},s_\mathrm{init},i$) \Comment{States to search from and too}
\For {$s \in I$}
\State $(c_d, z_d)_{d \in D} \gets \mathrm{Dijkstra}(s, D, \bar{U}_{i})$ \Comment{Search from $s$ to all $d \in D$ in $\bar{U}_i$}
\State $G \gets$ Add\_arcs($G,s,\bar{Z},i, (c_d, z_d)_{d \in D},\tilde{\psi}$) \Comment{Add arcs to $G$}
\EndFor
\EndFor
\State $(c_1,z_1) \gets \mathrm{Dijkstra}(s_\mathrm{init}, \bar{B}, G)$ \Comment{Get optimal plan $z$ to $(\bar{Z},s_\mathrm{init},\mathrm{start}(\bar{B}))$} 
\State \textbf{Part 2: Solve $(\bar{Z},\mathrm{start}(\bar{B}),s_\mathrm{goal})$}
\State Let $z_2$ be an empty trajectory \Comment{To be constructed}
\State $(M,c) \gets (\bar{D}_{k_1},\mathrm{start}(\bar{B}))$ 
\State $g \gets$ corresponding state in $M$ to $\bar{D}_{k_1-1}$ 
\While {$c \neq s_\mathrm{goal}$}
\State $z \gets \mathrm{Dijkstra}(c,g,M)$ \Comment{Search from $c$ to $g$ in $M$}
\State $z_2 \gets z_2 z$ \Comment{Add $z$ to $z_2$}
\State $c \gets start(g)$
\If {$c \neq s_\mathrm{goal}$}
\State Let $\bar{D}_{k_i}$ be the MM with state $c$
\State $M \gets \bar{D}_{k_i}$
\State $g \gets$ corresponding state in $M$ to $\bar{D}_{k_i-1}$ 
\EndIf
\EndWhile \Comment{End of Part 2}
\State return $z \gets z_1 z_2$
\end{algorithmic}
\end{algorithm}

\begin{algorithm}
\caption{Compute paths}\label{alg:compute_alpha_and_beta}
\begin{algorithmic}[1]
\Require $(Z,s_\mathrm{init},s_\mathrm{goal})$
\Ensure $(U_{path}, D_{path}, \alpha, \beta)$
\State Let $U_{path}$ be an empty sequence \Comment{Compute the path of $U$s}
\State $k \gets 0$
\State $U_0 \gets s_\mathrm{init}$
\State $U_{current} \gets U_0$ 
\State Add $U_0$ to $U_{path}$
\While {true}
\State $U_{k+1} \gets parent(U_{current})$
\If {$U_{k+1}$ not empty}
\State Add $U_{k+1}$ to $U_{path}$
\State $k \gets k+1$
\State $U_{current} \gets U_k$
\Else
\State \textbf{break}
\EndIf
\EndWhile
\State Let $D_{path}$ be an empty sequence \Comment{Compute the path of $D$s}
\State $k \gets 0$
\State $D_0 \gets s_\mathrm{goal}$
\State $D_{current} \gets D_0$
\State Add $D_0$ to $D_{path}$
\While {true}
\State $D_{k+1} \gets parent(D_{current})$
\If {$D_{k+1}$ not empty}
\State Add $D_{k+1}$ to $D_{path}$
\State $k \gets k+1$
\State $D_{current} \gets D_k$
\Else
\State \textbf{break}
\EndIf
\EndWhile
\State $(i,j) \gets (\mathrm{length}(U_{path}),\mathrm{length}(D_{path}))$ \Comment{Calculate $\alpha$ and $\beta$}
\While {$U_{path}[i] = D_{path}[j]$} 
\State $(i,j) \gets (i-1,j-1)$
\EndWhile
\State $(\alpha,\beta) \gets (i-1,j-1)$
\State {return} $(U_{path}, D_{path}, \alpha, \beta)$
\end{algorithmic}
\end{algorithm}

\begin{algorithm}[h]
\caption{Reduce}\label{alg:reduce}
\begin{algorithmic}[1]
\Require $(Z,s_{\mathrm{init}},s_{\mathrm{goal}})$ and $(c_x^M,z_x^M)_{x \in \Sigma, M \in X})$
\Ensure $(\bar{Z},U_{path},D_{path},\alpha,\beta)$
\State $(U_{path}, D_{path}, \alpha, \beta) \gets$ Compute\_paths($Z,s_\mathrm{init},s_\mathrm{goal}$)
\State Let $\bar{X}$ be an empty set \Comment{To be constructed}
\State Let $\bar{T}$ be an empty tree \Comment{To be constructed}
\State Construct $\bar{U}_1$ from $U_1$ and add to $\bar{X}$
\State Add node $\bar{U}_1$ to $\bar{T}$ with arcs out as $U_1$ in $T$ but no destination nodes
\State Save $q$ where $U_2 \xrightarrow{q} U_1 \in T$
\For {$i = 2,\dots,length(U_{path})-1$}
\State Construct $\bar{U}_i$ from $U_i$ and add to $\bar{X}$
\State Add node $\bar{U}_i$ to $\bar{T}$ with arcs out as $U_i$ in $T$ but no destination nodes
\State Add destination node $\bar{U}_{i-1}$ to $\bar{U}_{i} \xrightarrow{q}$ so that $\bar{U}_{i} \xrightarrow{q}\bar{U}_{i-1} \in \bar{T}$
\State Save $q$ where $U_{i+1} \xrightarrow{q} U_i \in T$
\EndFor
\If {$\beta = 0$} \Comment{Then nothing more to do}
\State {return} $(\bar{Z} = (\bar{X},\bar{T}), U_{path})$
\EndIf 
\State Construct $\bar{D}_1$ from $D_1$ and add to $\bar{X}$
\State Add node $\bar{D}_1$ to $\bar{T}$ with arcs out as $D_1$ in $T$ but no destination nodes
\State Save $q$ where $D_2 \xrightarrow{q} D_1 \in T$
\For {$i = 2,\dots,\beta$}
\State Construct $\bar{D}_i$ from $D_i$ and add to $\bar{X}$
\State Add node $\bar{D}_i$ to $\bar{T}$ with arcs out as $D_i$ in $T$ but no destination nodes
\State Add destination node $\bar{D}_{i-1}$ to $\bar{D}_{i} \xrightarrow{q}$ so that $\bar{D}_{i} \xrightarrow{q}\bar{D}_{i-1} \in \bar{T}$
\State Save $q$ where $D_{i+1} \xrightarrow{q} D_i \in T$
\EndFor
\State Add destination node $\bar{D}_{\beta}$ to $\bar{U}_{\alpha+1} \xrightarrow{q}$ so that $\bar{U}_{\alpha+1} \xrightarrow{q} \bar{D}_{\beta} \in \bar{T}$
\State {return} $(\bar{Z} = (\bar{X},\bar{T}),U_{path})$
\end{algorithmic}
\end{algorithm}

\begin{algorithm}[h] 
\caption{Compute\_transitions}\label{alg:compute_exit_transitions}
\begin{algorithmic}[1]
\Require $\bar{Z}$ and $U_{path}$ 
\Ensure Function $\tilde{\psi}$ prescribing: where one goes if exiting $\bar{U}_i$ with $x \in \Sigma$ (treating $\bar{B}$ as a state), saved as $\tilde{\psi}(\bar{U}_{i-1},x)$; or where one goes if entering $\bar{U}_i$ (treating $\bar{B}$ as a state), saved as $\tilde{\psi}(\bar{U}_{i-1})$.
\For {$x \in \Sigma$} \Comment{Compute exit transitions}
\State $exit(x) \gets \emptyset$
\EndFor
\For {$i = n,n-1,\dots,2$}
\For {$x \in \Sigma$}
\If {$\delta_{\bar{U}_i}(\bar{U}_{i-1},x) \neq \emptyset$}
\State $exit(x) \gets \delta_{\bar{U}_i}(\bar{U}_{i-1},x)$
\EndIf
\State $\tilde{\psi}(\bar{U}_{i-1},x) \gets exit(x)$
\EndFor
\EndFor
\State $go\_here \gets s_\mathrm{init}$ \Comment{Compute enter transitions}
\State $\tilde{\psi}(\bar{U}_0) \gets s_\mathrm{init}$
\For {$i = 1,\dots,n-1$}
\If {$s(\bar{U}_i) \neq \bar{U}_{i-1}$}
\State $go\_here \gets s(\bar{U}_i)$
\EndIf
\State $\tilde{\psi}(\bar{U}_i) \gets go\_here$
\EndFor
\State {return} $\tilde{\psi}$
\end{algorithmic}
\end{algorithm}

\begin{algorithm}[h]
\caption{Add\_nodes}\label{alg:add_nodes}
\begin{algorithmic}[1]
\Require $\bar{Z},U_{path},D_{path},\beta$ and empty graph $G$. 
\Ensure Graph $G$ with added nodes.
\State Add nodes $s_\mathrm{init}$, $s(\bar{U}_1)$ and $\bar{B}$ to $G$. 
\For {$i =2,\dots,n$}
\If {$s(\bar{U}_i) \neq \bar{U}_{i-1}$}
\State Add node $s(\bar{U}_i)$ to $G$
\EndIf
\For {$x \in \Sigma$} 
\If {$\delta_{\bar{U}_i}(\bar{U}_{i-1},x) \neq \emptyset$} 
\State Add node $\delta_{\bar{U}_i}(\bar{U}_{i-1},x)$ to $G$
\EndIf
\EndFor
\EndFor
\State return $G$
\end{algorithmic}
\end{algorithm}

\begin{algorithm}[h]
\caption{Get\_search\_states}\label{alg:get_search_states}
\begin{algorithmic}[1]
\Require $\bar{Z},s_\mathrm{init},i$
\Ensure States to search from $I$ and states to search too $D$ in $\bar{U}_i$
\State Get\_search\_states($\bar{Z},s_\mathrm{init},i$):
\State Let $I$ be an empty set \Comment{States we should search from}
\If {$s(\bar{U}_i) \neq \bar{U}_{i-1}$}
\State Add $s(\bar{U}_i)$ to $I$.
\EndIf
\If {$i=1$}
\State Add $s_\mathrm{init}$ to $I$
\Else
\For {$x \in \Sigma$} 
\If {$\delta_{\bar{U}_i}(\bar{U}_{i-1},x) \neq \emptyset$} 
\State Add $\delta_{\bar{U}_i}(\bar{U}_{i-1},x)$ to $I$
\EndIf
\EndFor
\EndIf
\State Let $D$ be an empty set \Comment{States we should search to}
\If {$i > 1$} 
\State Add $\bar{U}_{i-1}$ to $D$
\EndIf
\If {$i < n$} \Comment{Add exits as destinations}
\State Add $\{\mathrm{exit}_x\}_{x \in \Sigma}$ to $D$ 
\EndIf
\If {$i = \alpha+1$} 
\State Add $\bar{B}$ to $D$
\EndIf
\State {return} $(I,D)$
\end{algorithmic}
\end{algorithm}

\begin{algorithm}[h]
\caption{Add\_arcs}\label{alg:add_arcs}
\begin{algorithmic}[1]
\Require $G,s,\bar{Z},i, (c_d, z_d)_{d \in D},\tilde{\psi}$
\Ensure Graph $G$ with added arcs.
\If {$i>1$ and $c_{\bar{U}_{i-1}} < \infty$} \Comment{Add arcs to $G$}

\State Add arc $s \rightarrow \tilde{\psi}(\bar{U}_{i-1})$ labelled $(c_{\bar{U}_{i-1}}, z_{\bar{U}_{i-1}})$ 
\EndIf
\For {$\mathrm{exit}_x \in D$}
\If {$\tilde{\psi}(\bar{U}_{i},x) \neq \emptyset$ and $c_{\mathrm{exit}_x} < \infty$} \Comment{Transition in $\bar{Z}$ exists} 
\State Add arc $s \rightarrow \tilde{\psi}(\bar{U}_{i},x)$ to $G$ labelled $(c_{\mathrm{exit}_x},z_{\mathrm{exit}_x})$ 
\EndIf
\EndFor
\If {$i = \alpha+1$ and $c_{\bar{B}} < \infty$} 
\State Add arc $s \rightarrow \bar{B}$ to $G$ labelled $(c_{\bar{B}}, z_{\bar{B}})$
\EndIf
\State {return} $G$
\end{algorithmic}
\end{algorithm}

We first calculate the time complexity of all the algorithms that it calls, stated as lemmas.

\begin{lemma}\label{lemma:alg4}
Algorithm \ref{alg:compute_alpha_and_beta} has time complexity $O(\mathrm{depth}(Z))$.
\end{lemma}

\begin{lemma}\label{lemma:alg5}
Algorithm \ref{alg:reduce} has time complexity $O(\mathrm{depth}(Z) |b_s|)$.
\end{lemma}

\begin{lemma}\label{lemma:alg6}
Algorithm \ref{alg:compute_exit_transitions} has time complexity $O(|\Sigma| \mathrm{depth}(Z))$.
\end{lemma}

\begin{lemma}\label{lemma:alg7}
Algorithm \ref{alg:add_nodes} has time complexity $O(|\Sigma| \mathrm{depth}(Z))$.
\end{lemma}

\begin{lemma}\label{lemma:alg8}
Algorithm \ref{alg:get_search_states} has time complexity $O(|\Sigma|)$.
\end{lemma}

\begin{lemma}\label{lemma:alg9}
Algorithm \ref{alg:add_arcs} has time complexity $O(|\Sigma|)$.
\end{lemma}

The proof of the Lemmas \ref{lemma:alg4}, \ref{lemma:alg6}, \ref{lemma:alg7}, \ref{lemma:alg8}, and \ref{lemma:alg9} are straightforward by just going through all the steps in the corresponding algorithm. The proof of Lemma \ref{lemma:alg5} demands a bit more reasoning and we therefore provide a proof.

\begin{proof}[Proof of Lemma \ref{lemma:alg5}]
We walk through the algorithm. Line 1 has time complexity $O(\mathrm{depth}(Z))$ according to Lemma \ref{lemma:alg4}. Line 2 and 3 has time complexity $O(1)$. Concerning line 4, constructing $\bar{U}_i$ from $U_i$ for any $i$ be done in constant time $O(1)$, since the only change in $\bar{U}_i$ is $\bar{\gamma}$, which can be assigned in constant time using Equation \eqref{eq:bar_gamma} and storing $\gamma$ as a reference. Therefore, line 4 and 8 has time complexity $O(1)$.

\begin{remark}\label{remark_computing_bar_gamma}
Furthermore, note that $\bar{\gamma}$ can also be computed in constant time $O(1)$ from the reference $\gamma$ using Equation \eqref{eq:bar_gamma} (which is useful later in the search part to avoid additional complexities when evaluating $\bar{\gamma}$). To see this, note that checking the condition $\delta(q,x) \neq \emptyset$ in Equation \eqref{eq:bar_gamma} has time complexity $O(1)$. Also, checking the condition ${U}_i \xrightarrow{q} M \in {T}, M \in \M$ translates to checking if there exists $M \in \{ U_{i-1}, D_\beta \}$ such that ${U}_i \xrightarrow{q} M \in {T}$, which can be done in constant time. Hence, given $U_{path}$, $D_{path}$, $\beta$, all values of $\gamma$ and all the values $c^q_x$ obtained from the offline step, computing a value of $\bar{\gamma}$ has time complexity $O(1)$.
\end{remark}
Line 5 and 9 has time complexity $O(|b_s|)$ by going through all arcs. Line 6 and 11 has time complexity $O(1)$. Line 10 has time complexity $O(1)$. The for-loop in line 7-12 will be iterated at most $\mathrm{depth}(Z)$ times, and hence, the whole for-loop has time complexity $O(\mathrm{depth}(Z) |b_s|)$. We conclude that line 4-12 has time complexity $O(\mathrm{depth}(Z) |b_s|)$. Line 13-15 has time complexity $O(1)$. Line 16-24 has time complexity $O(\mathrm{depth}(Z) |b_s|)$ analogous to line 4-12. Finally, line 25-26 has time complexity $O(1)$. Summing all the complexities up, we get that Algorithm \ref{alg:reduce} has time complexity $O(\mathrm{depth}(Z) |b_s|)$.
\end{proof}
\noindent
We are now in a position to prove the time complexity of Part 1 of Algorithm \ref{alg:Step1_detailed}:
\begin{proof}
We walk through Part 1 of of Algorithm~\ref{alg:Step1_detailed}. Line 2 has time complexity $O(b_s \mathrm{depth}(Z))$ according to Lemma \ref{lemma:alg5}. Line 3 has time complexity $O(|\Sigma| \mathrm{depth}(Z))$ according to Lemma \ref{lemma:alg6}. Line 4 has time complexity $O(1)$. Line 5 has time complexity $O(|\Sigma| \mathrm{depth}(Z))$ according to Lemma \ref{lemma:alg7}. Line 7 has time complexity $O(|\Sigma|)$ according to Lemma \ref{lemma:alg8}. Concerning line 9, Dijkstra's algorithm has time complexity $O(E+V \log V)$ where $E$ is the number of arcs and $V$ is the number of (graph) nodes.\footnote{Here, we also use that $\bar{\gamma}$ of $\bar{U}_i$ can be computed in constant time, see Remark \ref{remark_computing_bar_gamma}.} In our case, we have at most $b_s+| \Sigma |$ (graph) nodes since we have at most $b_s$ (graph) nodes from $\bar{U}_i$ and then a maximum of $| \Sigma |$ additional (graph) nodes corresponding to all the possible exits.\footnote{Constructing the additional arcs to these exit-nodes has complexity $O(b_s |\Sigma|)$ and will not affect the overall complexity of line 9.} The maximum number of arcs are $b_s |\Sigma|$. Therefore, line 9 has time complexity $O(|\Sigma| b_s + (b_s+|\Sigma|) \log (b_s+|\Sigma|))$. Line 10 has time complexity $O(|\Sigma|)$ by Lemma \ref{lemma:alg9}. Let's now analyse line 8-11. The maximum number of elements in $I$ is $|\Sigma|+1$, and hence, line 8-11 has time complexity
\begin{align*}
O \big ( (|\Sigma|+1) \cdot [O(|\Sigma| b_s + (b_s+|\Sigma|) \log (b_s+|\Sigma|))+O(|\Sigma|)] \big ) \\ =  O \big (|\Sigma|^2 b_s + |\Sigma| (b_s+|\Sigma|) \log(b_s+|\Sigma|) \big ) =: H
\end{align*} 
Therefore, line 6-12 has time complexity
\begin{align*}
O(\mathrm{depth}(Z) [O(|\Sigma|)+H]) = \\ O \Big ( \big (|\Sigma|^2 b_s + |\Sigma| (b_s+|\Sigma|) \log(b_s+|\Sigma|) \big ) \mathrm{depth}(Z) \Big ).
\end{align*}
We continue with line 13. Note that $G$ has at most ${2+(n-1)[|\Sigma|+1]+1}$ (graph) nodes\footnote{Here, 2 is the number of (graph) nodes from $\bar{U}_1$, $(n-1)[|\Sigma|+1]$ is the number of (graph) nodes from the remaining $\bar{U}_i$, and 1 is to account for the (graph) node $\bar{B}$.}, which can be bounded by $C |\Sigma| \depth(Z)$ for some constant $C>0$. Denoting the number of (graph) nodes in $G$ by $V$, we then have
\begin{align*}
V \log (V) \leq C |\Sigma| \depth(Z) \log \big ( C |\Sigma| \depth(Z) \big ) = \\
C |\Sigma| \depth(Z) \big ( \log(C) + \log(|\Sigma| \depth(Z)) \big ) \leq \\
\tilde{C} |\Sigma| \depth(Z) \log(|\Sigma| \depth(Z))
\end{align*}
for some constant $\tilde{C}$ provided $|\Sigma|>1$ or $\depth(Z) >1$. We conclude that $V \log (V) = O(|\Sigma| \depth(Z) \log(|\Sigma| \depth(Z)))$. Furthermore, the maximum number of arcs in $G$ are
\begin{align*}
2 |\Sigma|+\max (n-2,0) \cdot (|\Sigma|+1)^2+(|\Sigma|+1)\cdot 1+(|\Sigma|+1) \\ = O( |\Sigma|^2 \depth(Z)).
\end{align*}
Here, the first term $2 |\Sigma|$ comes from all the arcs from $\bar{U}_1$ (2 graph nodes with $|\Sigma|$ arcs each), the second term comes from all the arcs from all $\bar{U}_i$ with $1<i<n$ ($|\Sigma|+1$ graph nodes with $|\Sigma|+1$ arcs each corresponding to all exits and going to $\bar{U}_{i-1}$), the third term comes from $\bar{U}_n$ ($|\Sigma|+1$ graph nodes with arcs going to $\bar{U}_{n-1}$), and the last term $(|\Sigma|+1)$ comes from the $\bar{U}_{i}$ that has $\bar{B}$ as a state ($|\Sigma|+1$ graph nodes that could all have 1 arc each to $\bar{B}$). Therefore, by the time complexity of Dijkstra's algorithm, line 13 has time complexity
\begin{equation*}
O( |\Sigma|^2 \depth(Z)+|\Sigma| \depth(Z) \log(|\Sigma| \depth(Z))).
\end{equation*}
By above, we conclude that Part 1 of Algorithm \ref{alg:Step1_detailed} has time complexity 
\begin{equation*}
O( |\Sigma|^2 \depth(Z))+O(|\Sigma| \depth(Z) \log(|\Sigma| \depth(Z))).
\end{equation*}
This completes the proof.
\end{proof}
We continue by proving the time complexity of Part 2 of Algorithm~\ref{alg:Step1_detailed}.
\begin{proof}
We walk through Part 2 of of Algorithm~\ref{alg:Step1_detailed}. Line 15 has time complexity $O(1)$. Concerning line 16, provided we have $D_{path}$ and $\beta$, then one can directly get $\bar{B}$ and in turn $\mathrm{start}(\bar{B})$ by at most $\depth(Z)$ recursive calls starting from $\bar{B}$. Hence, the time complexity of line 16 is $O(\depth(Z))$. Line 17 has time complexity $O(1)$ provided we have $D_{path}$. Concerning line 18-27, the number of function calls that line 20 needs to do over the whole while-loop is $O(\depth(Z))$. The same is true for line 21. Line 22-25 has time complexity $O(1)$, hence contributes also with $O(\depth(Z))$ over the whole while-loop. Furthermore, evaluating line 19 has time complexity $O(E+V\log V) = O(b_s |\Sigma| + b_s \log b_s)$ (where $E$ is the number of arcs and $V$ is the number of (graph) nodes in the graph used in Dijkstra's algorithm), and hence, contributes with $O([b_s |\Sigma| + b_s \log b_s] \depth(Z))$ over the whole while-loop. We conclude that line 18-27 has time complexity $O([b_s |\Sigma| + b_s \log b_s] \depth(Z))$, and since line 28 has time complexity $O(1)$, we get that Part 2 of Algorithm \ref{alg:Step1_detailed} also has time complexity 
\begin{equation*}
O([b_s |\Sigma| + b_s \log b_s] \depth(Z)).
\end{equation*}
This completes the proof.
\end{proof}
\begin{proof}[Proof of Proposition \ref{prop:time_complexity_step_1_and_2}]
Proposition \ref{prop:time_complexity_step_1_and_2} follows by adding the time complexities of Part 1 and Part 2 of Algorithm~\ref{alg:Step1_detailed}.
\end{proof}

\subsection{Proof of Proposition \ref{prop:time_complexity_step_3}}

\begin{proof}
Using Algorithm \ref{alg:step_3}, note that it takes time $O(\depth(Z))$ to obtain the next optimal input in the optimal plan $u$ to $(Z,s_\mathrm{init},s_\mathrm{goal})$, since $\mathrm{Plan\_expansion}$ operates in a depth-first-search manner. Therefore, we can also bound the time complexity for saving the whole optimal plan $u$ by $O(\depth(Z) |u| )$.
\end{proof}

\fi

\end{document}